\documentclass{llncs}

\usepackage{float}
\usepackage{mathpartir}
\usepackage{bm}
\usepackage{soul}
\usepackage{xspace}
\usepackage{mathtools}
\usepackage{hyperref}
\usepackage{tikz-cd}
\usetikzlibrary{positioning, shapes, arrows, fit}
\usepackage[scaled=0.7]{beramono}
\usepackage{calrsfs}
\usepackage{subdepth}
\usepackage{dashundergaps}
\usepackage[final]{listings}
\usepackage{mleftright}
\usepackage{extarrows}
\usepackage{placeins}
\usepackage[a]{esvect}
\usepackage{colortbl}
\usepackage{tabularx}
\usepackage{multirow}
\usepackage{subcaption}
\usepackage{enumitem}
\usepackage{longtable}
\usepackage{pifont}
\usepackage{cancel}
\usepackage{wrapfig}
\usepackage{stmaryrd}

\usepackage{etoolbox}
\usepackage{biblatex}
\usepackage{amssymb}
\usepackage{booktabs}

\newcommand{\citet}[1]{\citeauthor{#1} \cite{#1}}
\newcommand{\citep}{\cite}

\newcommand*{\derivationWidth}{0.68\textwidth}
\input{tex-common/proof-helpers}
\newcommand*{\proofContext}[1]{\def\currentprefix{proof:#1}}
\newcommand*{\locallabel}[1]{\label{\currentprefix:#1}}
\newcommand*{\localref}[1]{\ref{\currentprefix:#1}}

\newcolumntype{L}[1]{>{\raggedright\let\newline\\\arraybackslash\hspace{0pt}}p{#1}}
\newcolumntype{C}[1]{>{\centering\arraybackslash}p{#1}}

\DeclareSymbolFont{bbsymbol}{U}{bbold}{m}{n}
\DeclareMathSymbol{\bbsemi}{\mathbin}{bbsymbol}{"3B}

\everymath{\displaystyle}

\allowdisplaybreaks

\hypersetup{ 
    colorlinks=false,
    pdfborder={0 0 0},
}

\definecolor{Yellow}{RGB}{255,255,0}
\definecolor{LightGreen}{RGB}{144,238,144}

\definecolor{GainsboroDarkened}{RGB}{211,211,211}
\definecolor{GainsboroFaded}{RGB}{232,232,232}

\newbool{anonymous}

\booltrue{anonymous}%

\newcommand*{\ttt}[1]{\texttt{#1}}
\newcommand*{\kw}[1]{{\text{\ttt{#1}}}} 

\newcommand*{\twoPrime}{{\prime\mkern-2.6mu\prime\mkern-2.2mu}}

\DeclareTextFontCommand{\textbfit}{%
  \fontseries\bfdefault 
  \itshape
}

\newcommand*{\after}{\circ}

\newcommand*{\disjunion}{\uplus}

\newcommand*{\concat}{\cdot}

\newcommand*{\dom}[1]{\textsf{dom}(#1)}

\newcommand*{\eqdef}{\stackrel{\smash{\text{\tiny def}}}{=}}

\newcommand*{\Powerset}{\mathbb{P}}
\newcommand*{\powerset}[1]{\Powerset(#1)}

\newcommand*{\set}[1]{\{#1\}}

\newcommand*{\join}{\sqcup}
\newcommand*{\meet}{\sqcap}

\newcommand*{\param}{\cdot}

\newcommand*{\defref}[1]{Definition~\ref{def:#1}}

\newcommand*{\figref}[1]{Figure~\ref{fig:#1}}
\newcommand*{\figrefTwo}[2]{Figures \ref{fig:#1} and \ref{fig:#2}}

\newcommand*{\lemref}[1]{Lemma~\ref{lem:#1}}

\newcommand*{\propref}[1]{Proposition~\ref{prop:#1}}

\newcommand*{\proprefTwo}[2]{Propositions~\ref{prop:#1} and \ref{prop:#2}}

\newcommand*{\secref}[1]{Section~\ref{sec:#1}}
\newcommand*{\secrefTwo}[2]{Sections~\ref{sec:#1} and \ref{sec:#2}}

\newenvironment{nop}{}{}
\newenvironment{sdisplaymath}
   {\begin{nop}\small\begin{displaymath}}
   {\end{displaymath}\end{nop}\ignorespacesafterend}
\newenvironment{smathpar}
   {\begin{nop}\small\begin{mathpar}}
   {\end{mathpar}\end{nop}\ignorespacesafterend}

\newenvironment{mathfig}{\begin{sdisplaymath}}{\end{sdisplaymath}}

\makeatletter
\newbox\sf@box
  {\def\sf@one{#1}%
   \def\sf@two{#2}%
   \setbox\sf@box\hbox
      \bgroup}%
  { \egroup
   \ifx\@empty\sf@two\@empty\relax
     \def\sf@two{\@empty}
   \fi
   \ifx\@empty\sf@one\@empty\relax
      \subfloat[\sf@two]{\box\sf@box}%
   \else
      \subfloat[\sf@one][\sf@two]{\box\sf@box}%
   \fi}
\makeatother

\RequirePackage{xcolor}

\definecolor{highlightcolor}{rgb}{1.0,0.8,0.8}
\definecolor{shadecolor}{rgb}{0.9,0.9,0.9}
\definecolor{lightgray}{rgb}{0.8,0.8,0.8}
\newcommand*{\shadebox}[1]{\fcolorbox{lightgray}{shadecolor}{\raisebox{0pt}[0.60\baselineskip][0.05\baselineskip]{#1}}}

\newcommand*{\ruleName}[1]{\textnormal{\textsf{#1}}}

\makeatletter
\newcommand{\superimpose}[2]{%
  {\ooalign{$#1\@firstoftwo#2$\cr\hfil$#1\@secondoftwo#2$\hfil\cr}}}
\makeatother

\newsavebox{\vardisplaymathbox}

\definecolor{verylightgray}{gray}{0.9}
\definecolor{lightgray}{gray}{0.5}
\definecolor{mediumgray}{gray}{0.45}
\newlength\lsthorizontalpadding
\newcommand*\lstnumberstyle{\ttfamily\scriptsize\textcolor{lightgray}}
\newcommand*{\lbbar}{\{\kern-0.3em|}
\newcommand*{\rbbar}{|\kern-0.3em\}}
\newlength\lstnumbersep
\newlength\lstnumberwidth
\lstdefinelanguage{Fluid}{%
   morekeywords={as,else,fun,if,in,let,match,then}%
  ,moredelim=[s][\itshape]{`}{`}
}
\lstnewenvironment{codecol}[1][]{\lstset{language=code,#1,moredelim=*[is][\textcolor{mediumgray}]{|}{|}}}{}

\renewcommand*{\secref}[1]{\S~\ref{sec:#1}}
\newcommand*{\apdxsecref}[1]{%
   \ifappendices%
      \secref{#1}%
   \else{the supplementary material}%
   \fi}

\newcommand*{\edgeMinus}{\setminus_{\mathsf{E}}}

\renewcommand*{\ruleName}[1]{\textcolor{darkgray}{\textnormal{\textsf{#1}}}}
\newcommand*{\Lowlight}[1]{\textcolor{gray}{#1}}

\renewcommand*{\join}{\vee}

\renewcommand*{\meet}{\wedge}

\newcommand*{\OurLang}{[OurLanguage]\xspace}

\newcommand*{\tableref}[1]{Table \ref{table:#1}}

\newcommand*{\inStar}[2]{\set{#1 \mapsto #2}}

\renewcommand{\emptyset}{\varnothing}

\newcommand*{\compl}[1]{#1^{c}}
\newcommand*{\dual}[1]{#1^{\circ}}

\newcommand*{\symEqual}{\mathrel{\kw{=}}}

\newcommand{\cons}{\cdot}
\renewcommand{\concat}{\mathbin{++}}
\newcommand*{\iter}{..}
\newcommand*{\length}[1]{|#1|}
\newcommand*{\mathSf}[1]{\textup{\textsf{#1}}} 
\newcommand{\opName}[1]{\mathSf{#1}}
\newcommand*{\seq}[1]{\vv{#1}}
\newcommand*{\seqRange}[2]{\seqRangeOp{#1}{#2}{,\,}}
\newcommand*{\seqRangeOp}[3]{{#1} #3 \iter #3 {#2}}
\newcommand*{\Set}[1]{\mathSf{#1}}

\newcommand{\seqEmpty}{\epsilon}
\newcommand{\bind}[2]{{#1}:{#2}}

\newcommand{\envEmpty}{\varnothing}

\newcommand*{\datatype}[1]{\Sigma(#1)}

\newcommand*{\elimmapsto}{\mapsto}
\newcommand*{\elimBind}[2]{{#1}\elimmapsto{#2}}
\newcommand*{\elimConstr}[1]{\set{#1}}
\newcommand*{\elimRecord}[2]{\elimBind{\exRecord{#1}}{#2}}
\newcommand*{\elimVar}[2]{\elimBind{#1}{#2}}

\newcommand*{\exApp}[2]{{#1}\,{#2}}

\newcommand*{\exClosure}[3]{\kw{cl}({#1},{#2},{#3})}

\newcommand*{\exConstr}[2]{#1(#2)}

\newcommand*{\exForeignApp}[2]{{#1}({#2})}
\newcommand*{\exFun}[1]{\lambda{#1}}

\newcommand*{\exInt}[1]{#1}
\newcommand*{\exLetRecElim}[2]{\exLetRecPiecewise{#1}{#2}}
\newcommand*{\exLetRecPiecewise}[2]{\kw{let}\;{#1}\;\kw{in}\;{#2}}

\newcommand*{\record}[1]{\{{#1}\}}
\newcommand*{\exRecord}[1]{\record{#1}}

\newcommand*{\exVar}[1]{#1}

\newcommand*{\exLet}[3]{\kw{let}\;{#1}\symEqual{#2}\;\kw{in}\;{#3}}
\newcommand*{\exLetRec}[2]{\kw{let}\;{#1}\;\kw{in}\;{#2}}
\newcommand*{\exRec}[1]{\{#1\}}
\newcommand*{\exRecProj}[2]{#1.#2}

\newcommand*{\annot}[2]{#1_{#2}}

\newcommand*{\annClosure}[4]{\annot{\exClosure{#1}{#2}{#3}}{#4}}

\newcommand*{\annConstr}[3]{\annot{\exConstr{#1}{#2}}{#3}}

\newcommand*{\annInt}[2]{\annot{\exInt{#1}}{#2}}

\newcommand*{\annRec}[2]{\annot{\exRec{#1}}{#2}}

\newcommand*{\inN}[1]{\opName{in}_{#1}}

\newcommand*{\inE}[1]{\opName{inE}_{#1}}
\newcommand*{\outE}[1]{\opName{outE}_{#1}}
\newcommand*{\E}{\opName{E}}
\newcommand*{\V}{\opName{V}}

\newcommand*{\sinks}[1]{\mathsf{T}(#1)}
\newcommand*{\sources}[1]{\mathsf{S}(#1)}

\newcommand*{\Match}[1]{\Set{Match}}

\newcommand*{\Nat}{\mathbb{N}}

\newcommand{\evalR}{\Rightarrow}
\newcommand*{\evalS}{\evalR}
\newcommand*{\evalSugS}{\evalS}
\newcommand*{\evalSugR}[1]{\evalSugS}

\newcommand*{\closeDefs}{\rightarrowtail}
\newcommand*{\match}{\rightsquigarrow}

\newcommand*{\demandsR}{\triangledown}
\newcommand*{\demandedByR}{\rotatebox[origin=c]{180}{$\triangledown$}}
\newcommand*{\sufficesR}{\blacktriangle}
\newcommand*{\preimageDualR}{\blacktriangledown} 
\newcommand*{\relOutputR}{{\demandedByR}\!{\demandsR}}
\newcommand*{\relInputR}{{\demandsR}\!{\demandedByR}}

\newcommand*{\demandsAlg}[1]{\mathrel{\opName{demands}_{#1}}}
\newcommand*{\demandedByV}[1]{\mathrel{\opName{demByV}_{#1}}}
\newcommand*{\demandedByAlg}[1]{\mathrel{\opName{demBy}_{#1}}}

\newcommand*{\sufficesE}[1]{\mathrel{\opName{suffE}_{#1}}}
\newcommand*{\sufficesAlg}[1]{\mathrel{\opName{suff}_{#1}}}

\newcommand*{\disjjoin}{\mathbin{\ooalign{$\join$\cr%
   \hfil\raise0.42ex\hbox{$\scriptscriptstyle+$}\hfil\cr}}}

\newcommand*{\numleq}{\leq}

\renewcommand{\eqdef}{:=}

\newcommand*{\raw}[1]{\bm{#1}}

\newcommand*{\interpret}[1]{\hat{#1}}

\newcommand*{\caseName}[1]{\textbf{#1}}

\newcommand*{\subcase}[1]{%
\textbf{\emph{Subcase}}\ #1.& %
}

\newcommand*{\opGraph}[1]{#1^{-1}}
\newcommand*{\fresh}[2]{#1 \notin \V(#2)}

\newcommand{\resultLessSquash}[2]{#1 \!\! & \!\! {\textcolor{gray}{($\pm #2$)}}}

\newif\ifappendices
\appendicestrue

\bibliography{tex-common/bib}

\begin{document}

\title{Cognacy Queries over Dependence Graphs for \\Transparent Visualisations}
\author{ %
   Joe Bond\inst{1} \and %
   Cristina David\inst{1} \and %
   Minh Nguyen\inst{1} \and %
   Dominic Orchard\inst{2,3} \and %
   Roly Perera\inst{3,1} %
}

\institute{
  University of Bristol, Bristol, UK \\
  \email{j.bond@bristol.ac.uk} \\
  \email{cristina.david@bristol.ac.uk} \\
  \email{min.nguyen@bristol.ac.uk}
  \and
  University of Kent, Canterbury, UK \\
  \email{d.a.orchard@kent.ac.uk}
  \and
  University of Cambridge, Cambridge, UK \\
  \email{roly.perera@cl.cam.ac.uk}
}

\maketitle
\begin{abstract}
Charts, figures, and text derived from data play an important role in decision making, from data-driven policy
development to day-to-day choices informed by online articles. Making sense of or fact-checking outputs means
understanding how they relate to the underlying data. Even for domain experts with access to the source code
and data sets, this poses a significant challenge. In this paper we introduce a new program analysis framework
which supports interactive exploration of fine-grained IO relationships directly through computed outputs,
making use of dynamic dependence graphs. This framework enables a novel notion in data provenance which we
call \emph{linked inputs}, a relation of mutual relevance or ``cognacy'' which arises between inputs when they
contribute to common features of the output. Queries of this form allow readers to ask questions like ``What
outputs use this data element, and what other data elements are used along with it?''. We show how
\citeauthor{jonsson51}'s concept of conjugate operators on Boolean algebras appropriately characterises the
notion of cognacy in a dependence graph, and give a procedure for computing linked inputs over such a graph.

\hspace{5mm} To demonstrate the approach in practice, we present a functional programming language called
\OurLang which automatically enriches visual outputs with interactions supporting linked inputs and similar
fine-grained queries. We show how to obtain more informative, contextualised queries via projection operators
which factor out data contributed by specific inputs or demand associated with specific outputs. We also show
how to derive a \emph{linked outputs} operator, capturing a dual cognacy relation between output features
whenever they compete for common input data. Composed with suitable projections, this recovers a prior
approach to linked outputs based on execution traces and a bidirectional interpreter. However, the dependence
graph approach presented in this paper performs better (computationally) on most examples than an
implementation based on execution traces and is significantly simpler to implement.
\end{abstract}

\section{Introduction: Towards Transparent Research Outputs}
\label{sec:introduction}

Whether formulating national policy or making day-to-day decisions about our own lives, we increasingly rely
on the charts, figures and text created by scientists and journalists. Interpreting these visual and textual
summaries is essential to making informed decisions.   However, most of the artefacts we encounter are
\begin{wrapfigure}{r}{0.53\textwidth}
   \centering
   \includegraphics[width=0.5\textwidth]{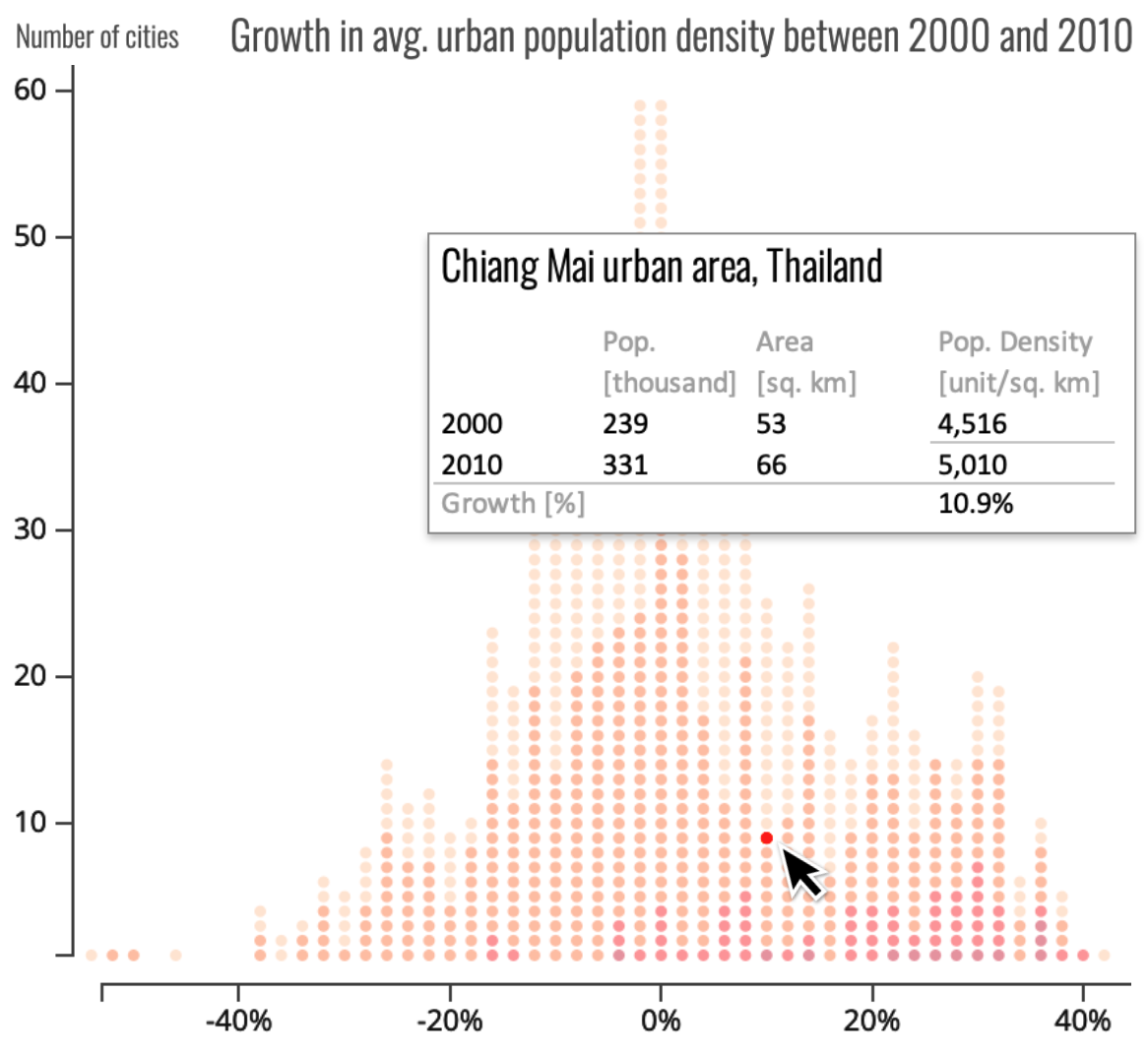}
   \caption{A hand-crafted transparent visualisation due to \citet{bremer15}}
   \vspace{-2mm}
   \label{fig:introduction:bremer}
\end{wrapfigure}
\noindent
\emph{opaque}, inasmuch as they are unable to reveal anything about how they relate to the data they were
derived from. Whilst one could in principle try to use the source code and data to reverse engineer some of
these relationships, this requires substantial expertise, as well as valuable time spent away from the
``comprehension context'' in which we encountered the output in question.  These difficulties are compounded
when the information presented draws on multiple data sources, such as medical
meta-analyses~\citep{dersimonian86}, ensemble models in climate science~\citep{murphy04}, or queries that span
multiple database tables~\citep{selinger79}. Even professional reviewers may lack the resources or inclination
to embark on such an activity. Perhaps more often than we would like, we end up taking things on trust.

With traditional print media, a ``disconnect'' between outputs and underlying data is unavoidable. For digital
media, other options are open to us. One way to improve things is to engineer visual artefacts to be more
``self-explanatory'' so they can reveal the relationship to the underlying data to an interested user.
Consider the histogram in \figref{introduction:bremer} showing urban population growth in Asia over a 10-year
period \cite{bremer15}. There are many questions a reader might have about what this chart represents --- in
other words, how visual elements map to underlying data. Do the points represent individual cities? What does
the colour scheme indicate? Which of the points represent large cities or small cities? Legends and other
accompanying text can help, but some ambiguities inevitably remain. The purpose of a visual summary after all
is to present the big picture at the expense of certain detail.

\citeauthor{bremer15}'s approach to this problem was to implement an interactive feature that allows a user to
explore some of these provenance-related questions \emph{in situ}, i.e.~directly from the chart. By selecting
an individual point, the user is able to bring up a view showing the data the point was calculated from. In
the figure the highlighted point represents Chiang Mai, and shows that the plotted value of 10.9\% was derived
from an increase in population density from 4,416 to 5,010 people per sq.~km. Features like these are valuable
as comprehension aids, but are laborious to implement by hand and require the author to anticipate the queries
a user might have. For these reasons they also tend not to generalise: for example \citeauthor{bremer15}'s
visualisation only allows the user to select one point at a time.

\subsection{Data transparency as PL infrastructure}
\label{sec:introduction:data-transparency}

Hand-crafted efforts like \citet{bremer15}'s are labour-intensive because they involve manually embedding
metadata about the relationship between outputs and inputs into the same program. When the visualisation or
analysis logic changes, the relationship between inputs and outputs also change, and the metadata must be
manually updated. A less costly and more robust approach is to treat data provenance as a language
infrastructure problem, baking lineage or provenance metadata directly into outputs so that provenance queries
can be supported automatically. For example, for an in-memory database engine, \citet{psallidas18smoke}
describe how to ``backward trace'' from output selections to input selections, and then ``forward trace'' to
find related output selections in other views, to support a popular feature from data visualisation called
\emph{linked brushing}. \citet{perera22} implemented a similar system for a general-purpose functional
programming language, using execution traces to support bidirectional queries.

The advantage of shifting the burden of implementing transparency features onto the language runtime is that
the author of the content can concentrate purely on data visualisation, and as the infrastructure improves,
the benefits are inherited automatically by end users, at no additional cost to the author. For example if
some kind of formal guarantees are provided (perhaps that data selections are in some sense minimal and
sufficient), then those can be proved once for the infrastructure rather for each bespoke implementation.

In this paper, we propose a new bidirectional analysis framework for a general-purpose programming language,
which supports fine-grained \emph{in situ} provenance queries, an end-user feature we call \emph{data
transparency}. In contrast to prior work, we use \emph{dynamic dependence graphs}~\cite{ferrante87, agrawal90}
to implement the analyses, which enables our approach to be fast enough for interactive use and also
language-independent, by separating queries over the graph from the problem of deriving the dependence graph
for a particular program.

\subsection{Contributions and roadmap}
\label{sec:introduction:contributions}

Our specific contributions are as follows:

\begin{itemize}[leftmargin=*]
   \item \secref{core} defines a core calculus for data-transparent outputs, where parts of inputs and
      computed values are assigned unique labels called \emph{addresses}, and programs have an operational
      semantics that pairs every result with a dynamic dependence graph capturing fine-grained IO
      relationships between input addresses and output addresses.

   \item \secref{conjugate} presents a new formal framework for bidirectional provenance queries over
      dependence graphs, formalising two operators over such graphs, $\demandsR$ (\emph{demands}) and
      $\demandedByR$ (\emph{demanded by}). We show these to be \emph{conjugate} in the sense
      of~\citet{jonsson51}, and give procedures for computing $\demandsR$ and $\demandedByR$. We show how a
      novel cognacy operator $\relInputR$ called \emph{linked inputs}, relating inputs when they contribute to
      common features of the output, can be obtained by composing $\demandsR$ and $\demandedByR$, and how a
      dual \emph{linked outputs} operator $\relOutputR$, supporting the ``linked visualisations'' of prior
      work \cite{psallidas18smoke,perera22}, is easily obtained by transposing the two operators.

   \item \secref{evaluation} shows that our implementation performs better than one based on traces and
      bidirectional interpreters, the primary alternative implementation technique. Using usability metrics
      from \citet{nielsen93}, we find that $81\%$ of the queries we tested execute at a speed that appears
      instantaneous to a user, compared to $25\%$ using an implementation based on traces. We also compare
      overhead of building a trace vs.~building a dependence graph for a given program.
\end{itemize}

\noindent These contributions are not available in previous work. In particular, although
\citet{psallidas18smoke} support fast queries for linked visualisations in a database setting, and
\citet{perera22} support linked outputs in a general-purpose languages, ours is (to the best of our knowledge)
the first implementation for a general-purpose language which is fast enough for interactive use. Other
related work in database provenance, program slicing and data visualisation are discussed in
\secref{related-work}.

\secref{conclusion} wraps up with a discussion of some limitations. In particular, the sort of transparency
considered in this paper is purely extensional, and falls short of providing full explanations of how output
parts are related to input parts. We discuss more intensional forms of transparency in \secref{conclusion} and
propose some other ways in which the present system could be improved.

We implement our approach in a pure functional programming language called \OurLang. The author of a
visualisation expresses their chart as a pure function of the inputs, using a set of built-in data types for
common visualisations; a d3.js front end automatically enriches the rendered outputs with support for
interactive selection and the data transparency queries introduced in the next section. Our implementation is
open source and available online (given later in a deanonymised version). We also provide an anonymised web
demo at \href{http://opencomputation.org/}{\textsf{http://opencomputation.org/}}.

\section{Overview: Fine-Grained Interactive Provenance}
\label{sec:overview}

In this section, we introduce the main interactive data provenance features that wish to support, using two
\OurLang examples to illustrate. The line chart in \figref{introduction:related-inputs} shows projected
methane emissions from agricultural sources under a global warming scenario called RCP8.5, with source code in
\figref{example:moving-average:source}; the scatter plot and stacked bar chart in
\figref{introduction:scatterplot} show changing non-renewable energy outputs and capacities for various
countries, with source code in \figref{example:non-renewables:source}.

The key idea of a transparent visualisation is that the original data sources are kept around and can be
viewed (when a user so requests) alongside the visualisation, and the user is then able to interact with both
the data and the view to explore how they are related. Crucially the author of the visualisation does not to
have to implement any of these interactions themselves; they need only write the visualisation code and the
language runtime and rendering infrastructure provides the interactions.

\paragraph{1. Fine-grained linking of the data to the view.}

In \figref{introduction:related-inputs}, the user has chosen to reveal the underlying data set, which is shown
on the left; only rows relevant to the chart are visible, with the other rows hidden automatically. This
clarifies that only agricultural data is relevant, confirming the (informal) claim in the caption. The
emissions values shown do contribute to the chart in some way, and the user can investigate this by
interacting with individual entries. For example, moving their mouse over the number 104.69 highlights
\emph{one} point in the projected emfissions curve, and \emph{three} points in the other curve, which plots the
moving average of the projected emissions. The provenance analysis which underpins this we write as
$\demandedByR$ (``demanded by''), and here this tells us that 104.69 was needed to compute either the $x$ or
$y$ coordinate (in this case just the $y$ coordinate) of the four highlighted points.

\begin{figure}
   \centering
   \includegraphics[width=0.95\textwidth]{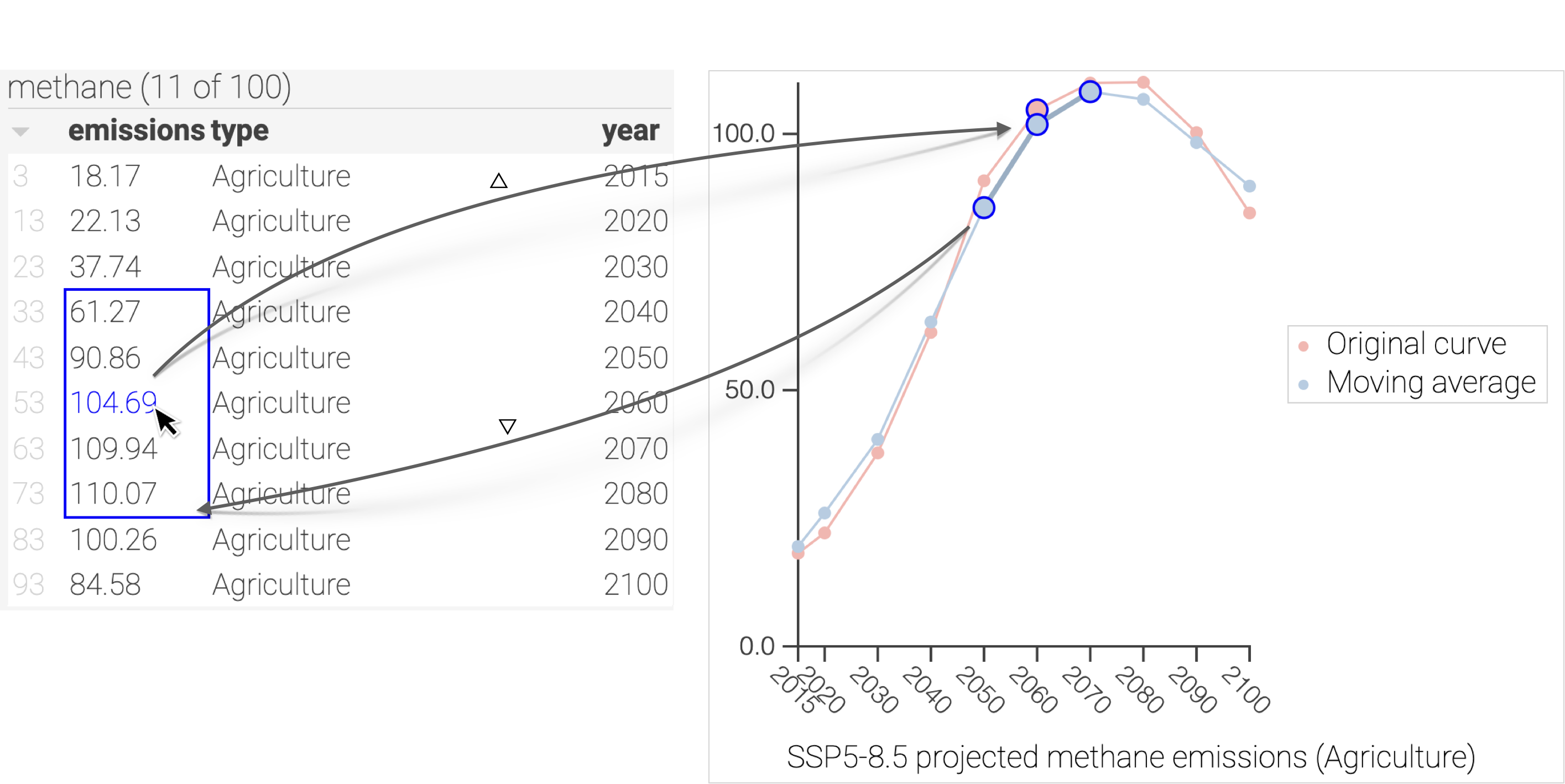}
   \caption{Data transparency: \emph{demanded by} ($\demandedByR$) and \emph{demands} ($\demandsR$) operators
   link inputs that share common output dependencies.}
   \label{fig:introduction:related-inputs}
\end{figure}

\begin{figure}
   \small
       {
\begin{lstlisting}[language=Fluid]
let nthPad n xs = nth (min (max n 0) (length xs - 1)) xs;
let movingAvg ys window =
      [ sum [ nthPad n ys | n $\leftarrow$ is ] / (1 + 2 $\ast$ window)
      | i $\leftarrow$ [ 0 .. length ys - 1 ],
        let is = [ i - window .. i + window ] ];
let movingAvg' rs window =
      zipWith (fun x y $\rightarrow$ {x: x, y: y})
         (map (fun r $\rightarrow$ r.x) rs)
         (movingAvg (map (fun r $\rightarrow$ r.y) rs) window);
let points =
      [ { x: r.year, y: r.emissions } | r $\leftarrow$ methane, r.type == "Agriculture" ]
in LineChart {
   tickLabels: { x: Rotated, y: Default },
   size: { width: 330, height: 285 },
   caption: "SSP5-8.5 projected methane emissions (Agriculture)",
   plots: [
      LinePlot { name: "Original curve", points: points },
      LinePlot { name: "Moving average", points: movingAvg' points 1 }
   ] }
\end{lstlisting}}
       \vspace{-0.5em}
   \caption{Moving average source code}
   \label{fig:example:moving-average:source}
\end{figure}

\paragraph{2. Linked inputs.}

We want to support fine-grained IO queries that run in the other direction too, via another provenance
analysis written $\demandsR$ (``demands''). In \figref{introduction:related-inputs} the $\demandsR$ analysis
is initiated automatically on the output of $\demandedByR$; this reveals that calculating the $y$ coordinates
of the points highlighted on the right required not only the 104.69 we started with, but 4 additional
emissions values (blue border on the left). These additional inputs we refer to as the \emph{related inputs}
of the original input selection; they are the other inputs demanded by any output that demands our starting
input, and are computed using the composite operator $\relInputR$. Given that the initial input selection here
contributes to 3 points of the moving average, the ``window'' of related inputs in this case is 5 wide,
comprising all the data points needed to account for the selection on the right. Related inputs is a
\emph{cognacy} relation: two inputs are related if they have a common ancestor in a \emph{dependence graph}
that captures how inputs are demanded by outputs.

\paragraph{3. Fine-grained linking of the view to the data.}

We would also like to support cognacy queries that start from the output rather than the input; this is the
basis of a feature called \emph{linked brushing} (or \emph{brushing and linking}) in data visualisation
\cite{becker87,buja91}. Here we would like to provide linked brushing in a way that is transparent to the
user. In \figref{introduction:scatterplot}, the user has again revealed the underlying data set, this time
opting to see only rows with active data selections, rather than only rows with data that are used by any part
of the output. They then express interest in one of the points in the scatter plot. By clicking on it, rather
than just moving their mouse over it, they create what we call a \emph{persistent} selection (shown in green).
The demands analysis $\demandsR$ reveals (also in green) the inputs needed to compute both the $x$ and $y$
coordinates of the selected point; this shows that the 2016 data for 4 countries was used, although it does
not reveal how.

\paragraph{4. Linked outputs.}

Now that the inputs demanded by the output selection are determined, the $\demandedByR$ analysis runs
automatically on the output of $\demandsR$; this reveals that those inputs were also needed for 4 of the bar
segments (highlighted with cross-hatching) in the bar chart, namely those for 2016. This is standard linked
brushing, and is implemented using the composite operator $\relOutputR$, which we call \emph{related outputs}.
Related outputs is a ``co-cognacy'' relation, relating two outputs whenever they have a common descendant,
rather than common ancestor, in the dependence graph. Here the user can take advantage of the fact that a
persistent selection remains active even when the mouse moves off the original item of interest; this allows
them to investigate further. Moving their mouse over the yellow IND segment of the 2016 bar (highlighted with
blue border) initiates an orthogonal $\demandsR$ query which shows (also with various blue borders) the data
needed to compute the yellow segment. Overlaid on the persistent selection, this reveals that the 37.90 in the
\kw{nuclearOut} column is needed to compute \emph{both} the IND bar segment \emph{and} the selected scatter
plot point, explaining why the two output selections are linked.

\begin{figure}[t]
   \centering
   \includegraphics[width=0.98\textwidth]{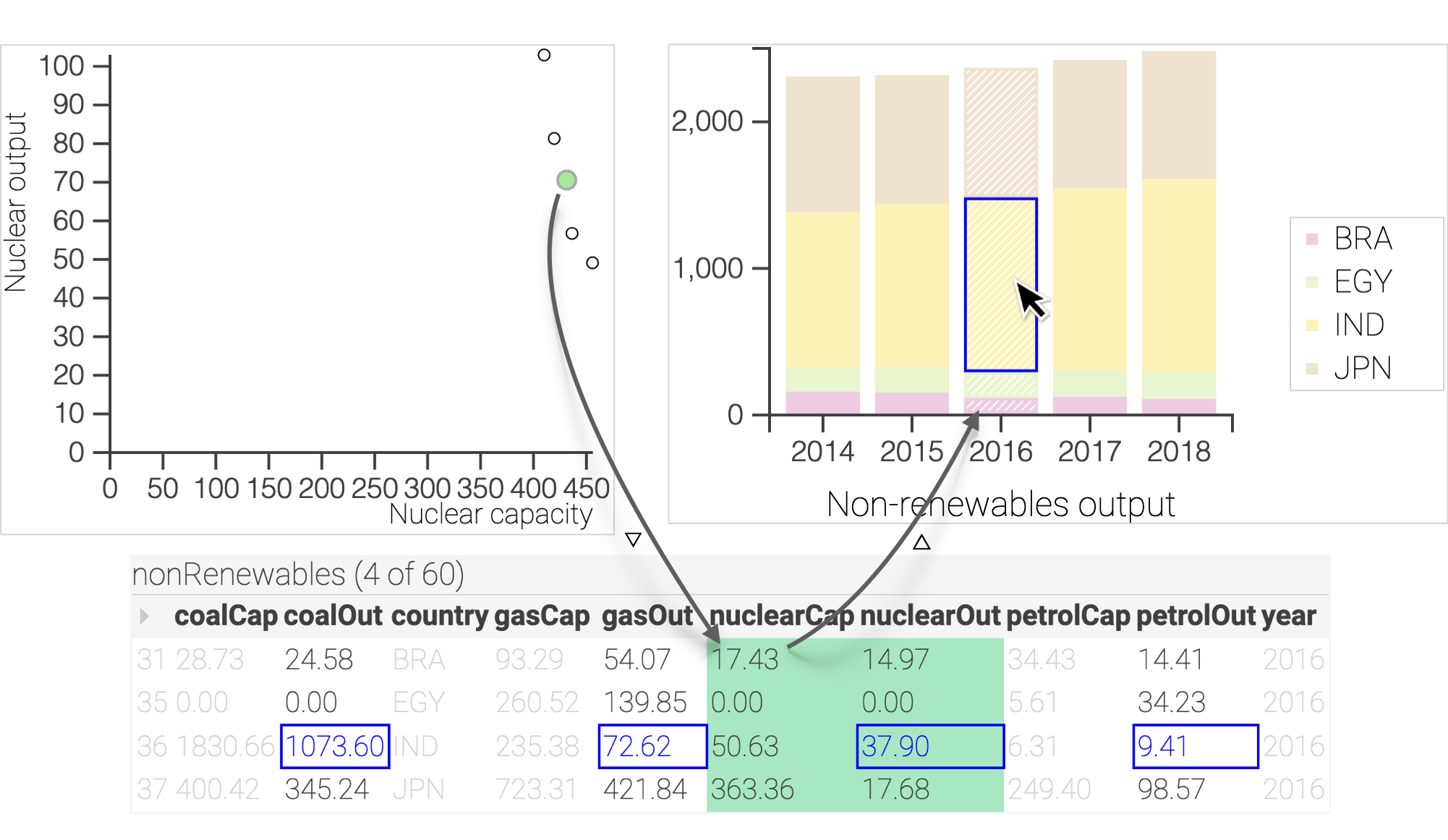}
   \caption{Data transparency: \emph{demands} ($\demandsR$) and \emph{demanded by} ($\demandedByR$)
   operators link outputs that share common input dependencies.}
   \label{fig:introduction:scatterplot}
\end{figure}

\begin{figure}[!h]
  \small
      {
\begin{lstlisting}[language=Fluid]
let countries = ["BRA", "EGY", "IND", "JPN"];
let totalFor year country =
   let [ row ] =
      [ row | row $\leftarrow$ nonRenewables, row.year == year, row.country == country ]
   in row.nuclearOut + row.gasOut + row.coalOut + row.petrolOut;
let stack year =
      [ { y: country, z: totalFor year country } | country $\leftarrow$ countries ];
let yearData year =
      [ row | row $\leftarrow$ nonRenewables,
              row.year == year, row.country `elem` countries ]
in MultiView {
   barChart: BarChart {
      caption: "Non-renewables output",
      size: { width: 275, height: 185 },
      stackedBars:
         [ { x: numToStr year, bars: stack year } | year $\leftarrow$ [2014..2018] ]
   },
   scatterPlot: ScatterPlot {
      caption: "",
      points:
        [ { x: sum [ row.nuclearCap | row $\leftarrow$ rows ],
            y: sum [ row.nuclearOut | row $\leftarrow$ rows ] }
        | year $\leftarrow$ [2014..2018], let rows = yearData year ],
      xlabel: "Nuclear capacity",
      ylabel: "Nuclear output"
   }
}
\end{lstlisting}}
      \vspace{-0.5em}
   \caption{Non-renewable energy source code}
   \label{fig:example:non-renewables:source}
\end{figure}


\section{A Core Calculus for Transparent Outputs}
\label{sec:core}

We start by introducing a calculus for transparent outputs, which is the core of \OurLang. The core language
is fairly expressive, supporting deep pattern-matching and mutual recursion; however \OurLang also provides a
surface syntax that desugars into the core, providing piecewise function definitions, list comprehensions and
other conveniences, as shown in \figrefTwo{example:moving-average:source}{example:non-renewables:source},
which would otherwise complicate the core.

The term syntax of the core language (\secref{core:syntax}) is uncontroversial: it is pure and untyped, and
provides datatypes, records and matrices (although we omit a presentation of matrices here). The syntax of
values is less standard: each part of a value has a unique \emph{address} $\alpha$ which serves to identify
that part of the value in a dependence graph. We give a big-step operational semantics that evaluates a term
to both a value and a \emph{dynamic dependence graph} for that value (\secref{core:semantics}), capturing how
parts of the value depend on parts of the input.

\paragraph{Some notation.} We write to ${\seq{x}}$ denote a finite sequence of elements $\seqRange{x_1}{x_n}$,
with $\seqEmpty$ as the empty sequence. Concatenation of sequences is written $\seq{x} \concat \seq{x}'$; we
also write $x \cons \seq{x}$ for cons (prepend) and $\seq{x} \cons x$ for snoc (append). We write
$\set{\seq{\bind{k}{x}}}$ to denote a finite map, i.e.~a set of pairs
$\seqRange{\bind{k_1}{x_1}}{\bind{k_n}{x_n}}$ where keys $k_i$ are pairwise unique. If $X$ and $Y$ are sets,
we write $X \disjunion Y$ to mean $X \cup Y$ where $X$ and $Y$ are disjoint, and also $x \cdot X$ or $X \cdot
x$ to mean $X \disjunion \set{x}$.

\subsection{Syntax}
\label{sec:core:syntax}

\begin{figure}[t]
   \small
   \begin{minipage}[t]{0.5\textwidth}
      \begin{tabularx}{\textwidth}{rL{2cm}L{3.6cm}}
         &\textit{Expression}&
         \\
         $e ::=$
         &
         $\exVar{x}$
         &
         variable
         \\
         &
         $\exInt{n}$
         &
         integer
         \\
         &
         $\exLet{x}{e}{e'}$
         &
         let
         \\
         &
         $\exRec{\seq{\bind{x}{e}}}$
         &
         record
         \\
         &
         $\exRecProj{e}{x}$
         &
         record projection
         \\
         &
         $\exConstr{c}{\seq{e}}$
         &
         constructor
         \\
         &
         $\exApp{e}{e'}$
         &
         application
         \\
         &
         $\exForeignApp{f}{\seq{e}}$
         &
         foreign application
         \\
         &
         $\exFun{\sigma}$
         &
         function
         \\
         &
         $\exLetRec{\rho}{e}$
         &
         recursive let
         \\[2mm]
         & \multicolumn{2}{l}{\textit{Continuation}}
         \\
         $\kappa ::=$
         &
         $e$
         &
         expression
         \\
         &
         $\sigma$
         &
         eliminator
         \\[9mm]
      \end{tabularx}
   \end{minipage}%
   \begin{minipage}[t]{0.48\textwidth}
      \begin{tabularx}{\textwidth}{rL{2.7cm}L{3cm}}
         &\textit{Eliminator}&
         \\
         $\sigma ::=$
         &
         $\elimVar{x}{\kappa}$
         &
         variable
         \\
         &
         $\elimRecord{\seq{x}}{\kappa}$
         &
         record
         \\
         &
         $\elimConstr{\seq{\elimBind{c}{\kappa}}}$
         &
         constructor
         \\[2mm]
         &\textit{Value}&
         \\
         $v ::=$
         &
         $\annot{\raw{v}}{\alpha}$
         &
         \\[2mm]
         &\textit{Raw value}&
         \\
         $\raw{v} ::=$
         &
         $\exInt{n}$
         &
         integer
         \\
         &
         $\exRec{\seq{\bind{x}{v}}}$
         &
         record
         \\
         &
         $\exConstr{c}{\seq{v}}$
         &
         constructor
         \\
         &
         $\exClosure{\gamma}{\rho}{\sigma}$
         &
         closure
         \\[2mm]
         &\textit{Environment} &
         \\
         $\gamma ::=$
         & $\set{\seq{\bind{x}{v}}}$
         \\[2mm]
         & \multicolumn{2}{l}{\textit{Recursive definitions}}
         \\
         $\rho ::=$
         & $\set{\seq{\bind{x}{\sigma}}} $
      \end{tabularx}
   \end{minipage}
   \caption{Syntax of the core language, including values labeled with addresses}
\label{fig:syntax}
\end{figure}

\subsubsection{Terms}
\label{sec:core:expressions}

\figref{syntax} defines expressions $e$ of the language, which include variables $x$, integer constants $n$,
let-bindings $\exLet{x}{e}{e'}$, record construction $\exRec{\seq{\bind{x}{e}}}$, record projection
$\exRecProj{e}{x}$, and (saturated) constructor expressions $\exConstr{c}{\seq{e}}$, where $c$ ranges over
data constructors. The language is parameterised by a finite map $\Sigma$ from constructors $c$ to constructor
arities $\datatype{c} \in \mathbb{N}$. Function application comes in two forms: the usual $\exApp{e}{e'}$, and
(saturated) \emph{foreign function} applications $\exForeignApp{f}{\seq{e}}$ (\secref{core:foreign} below).
The final two expression forms are anonymous functions $\exFun{\sigma}$, where $\sigma$ is a pattern-matching
construct called an \emph{eliminator} (\secref{core:continuations} below), and sets of mutually recursive
functions $\exLetRecElim{\rho}{e}$, where $\rho$ is a finite map $\seq{\bind{x}{\sigma}}$ from names to
eliminators.

\subsubsection{Foreign functions}
\label{sec:core:foreign}

The language is also parameterised by a finite map $\Phi$ from variables $f$ to arities $\Phi(f) \in \Nat$ of
the foreign function they denote. For the graph semantics in \secref{core:semantics} below, every foreign
function name $f$ is required to provide an interpretation $\interpret{f}$ that, for a sequence of arguments
$\seq{v}$ and dependence graph $G$, returns the result of applying the foreign function named by $f$ to
$\seq{v}$ plus a dependence graph $G'$ which extends $G$ with information about how the result depends on
$\seq{v}$.

\subsubsection{Continuations and eliminators}
\label{sec:core:continuations}

A \emph{continuation} $\kappa$ is a term $e$ or an eliminator $\sigma$, describing how an execution proceeds
after a value is matched. An \emph{eliminator} is a deep pattern-matching construct based on
tries~\cite{hinze00,peytonjones2021}; for a language with rich structured values like records and data types,
this makes for a cleaner presentation than a single shallow elimination form for each type. Piecewise function
definitions in the surface language desugar into eliminators.

An eliminator specifies how values of a particular shape are matched, and for a given value, determines how
any pattern variables get bound and the continuation $\kappa$ which will be executed under those bindings. A
variable eliminator $\elimVar{x}{\kappa}$ says how to match any value (as variable $x$) and continue as
$\kappa$. A record eliminator $\elimRecord{\seq{x}}{\kappa}$ says how to match a record with fields $\seq{x}$,
and provides a continuation $\kappa$ for sequentially matching the values of those fields. Lastly, a
constructor eliminator $\elimConstr{\seq{c \mapsto \kappa}}$ provides a branch $c \mapsto \kappa$ for each
constructor $c$ in $\seq{c}$, where each $\kappa$ specifies how any arguments to $c$ will be matched. (We
assume that any constructors matched by an eliminator all belong to the same data type, but this is not
enforced in the core language.)

\subsubsection{Addresses, values and environments}
\label{sec:core:values}

We define addressed \emph{values} $v$ mutually inductively with \emph{raw values}. A raw value is simply a
value without an associated address; a value decorates a raw value with an address $\alpha$. Addresses are
allocated during evaluation so that new partial values have fresh addresses, which can then be used as
vertices in a dependence graph.

Raw values $\raw{v}$ include integers $n, m$; records  $\exRecord{\seq{x \mapsto v}}$; and constructor values
$\exConstr{c}{\seq{v}}$. Raw values also include closures $\exClosure{\gamma}{\rho}{\sigma}$ where $\sigma$ is
the function body, $\gamma$ the captured environment, and to support mutual recursion, ${\rho}$ is the
(possibly empty) set of named functions with which $\sigma$ was mutually defined. \emph{Environments} are
finite maps from variables to values. Because foreign functions in the core language are not first-class and
calls $\exForeignApp{f}{\seq{e}}$ are saturated, i.e.~$\length{\seq{e}} = \Phi{(f)}$, the surface language
provides a top-level environment which maps every foreign function name ${f}$ of arity $n$ to the closure
$\annClosure{\envEmpty}{\envEmpty}{\elimVar{x_1}{\elimVar{\iter}{\elimVar{x_{n}}{\exForeignApp{f}{\seq{x}}}}}}{\alpha}$,
with $\alpha$ fresh, emulating first-class foreign functions.

\subsection{Operational Semantics}
\label{sec:core:semantics}

We now give a big-step operational semantics for the core language, which evaluates a term to a value paired
with a \emph{dynamic dependence graph} for that value.

\subsubsection{Dynamic Dependence Graphs}
\label{sec:core:dependence-graphs}

A \emph{dynamic dependence graph}~\cite{agrawal90} (hereafter \emph{dependence graph}) is a directed acyclic
graph $G = (V,E)$ with a set $V$ of \emph{vertices} and a set $E \subseteq V \times V$ of edges. When
convenient we write $\V(G)$ for $V$ and $\E(G)$ for $E$. In the dependence graph for a particular program,
vertices $\alpha, \beta \in V$ are addresses associated to values (either supplied to the program as inputs or
produced during evaluation) and edges $(\alpha, \beta) \in E$ indicate that, in the evaluation of that
program, the value associated to $\beta$ \textit{depends on} (is \emph{demanded by}, in the terminology of
\secref{overview}) the value associated to $\alpha$. Such values may be sub-terms of a larger value. This
diverges somewhat from traditional approaches to dynamic dependence graphs in only considering one type of
edge (data dependency) rather than separate data and control dependencies; moreover our edges point in the
direction of dependent vertices, whereas in the literature the other direction is somewhat more common.

During evaluation, when a fresh vertex $\alpha$ is allocated for a constructed value, the dependence graph $G$
is extended by a graph fragment specifying that $\alpha$ depends on a set $V$ of preexisting vertices, given
by the following notation:

\begin{definition}[In-star notation]
Write $\inStar{V}{\alpha}$ as shorthand for the star graph $(\set{\alpha} \disjunion V, V \times
\set{\alpha})$.
\end{definition}

\begin{figure}
   {\small \flushleft \shadebox{$\seq{v}, \kappa \match \gamma, e, V$}\hfill
   \begin{smathpar}
      \inferrule*[
         lab={\ruleName{$\match$-done}}
      ]
      {
         \strut
      }
      {
         \seqEmpty, e \match \envEmpty, e,  \emptyset
      }
      \and
      \inferrule*[lab={\ruleName{$\match$-var}}]
      {
         \seq{v}, \kappa \match \gamma, e, V
      }
      {
         v \cons \seq{v}, \elimVar{x}{\kappa}
         \match
         \gamma \cons (\bind{x}{v}), e, V
      }
      \\
      \inferrule*[
         lab={\ruleName{$\match$-record}}
      ]
      {
         \exRec{\seq{\bind{y}{u}}} \subseteq \exRec{\seq{\bind{x}{v}}}
         \\
         \seq{u} \concat \seq{v}', \kappa
         \match
         \gamma, e, V
      }
      {
         \annot{\exRec{\seq{\bind{x}{v}}}}{\alpha} \cons \seq{v}',
         \elimRecord{\seq{y}}{\kappa}
         \match
         \gamma, e, \alpha \cons V
      }
      \and
      \inferrule*[lab={\ruleName{$\match$-constr}}
                , right={$\Sigma(c) = |\seq{v}|$}]
      {
         \seq{v} \concat \seq{v}', \kappa
         \match
         \gamma, e, V
      }
      {
         \annot{\exConstr{c}{\seq{v}}}{\alpha} \cons \seq{v}', (\elimBind{c}{\kappa}) \cons \elimConstr{\seq{\elimBind{c}{\kappa}}}
         \match
         \gamma, e, \alpha \cons V
      }
   \end{smathpar}}
\vspace{-0.5cm}
\caption{Pattern-matching}
\label{fig:core:pattern-matching}
\end{figure}

\begin{figure}
  \vspace{-1em}
   {\small \flushleft \shadebox{$\gamma, e, V, G \evalS v, G'$}%
   \begin{smathpar}
      \inferrule*[
      lab={\ruleName{$\evalS$-var}}
      ]
      {
         \strut
      }
      {
         \gamma \cons (\bind{x}{v}), \exVar{x}, V, G
         \evalS
         v,
         G
      }
      \and
      \inferrule*[
         lab={\ruleName{$\evalS$-int}}
      ]
      {
         \fresh{\alpha}{G}
      }
      {
         \gamma, n, V, G
         \evalS
         \annInt{n}{\alpha},
         G \cup \inStar{V}{\alpha}
      }
      \and
      \inferrule*[lab={\ruleName{$\evalS$-function}}]
      {
         \fresh{\alpha}{G}
      }
      {
         \gamma, \exFun{\sigma}, V, G
         \evalS
         \annClosure{\gamma}{\envEmpty}{\sigma}{\alpha},
         G \cup \inStar{V}{\alpha}
      }
      \and
      \inferrule*[lab={\ruleName{$\evalS$-record}}]
      {
         \gamma, \seq{e}, V, G \evalS \seq{v}, G'
         \\
         \fresh{\alpha}{G'}
      }
      {
         \gamma, \exRec{\seq{\bind{x}{e}}}, V, G
         \evalS
         \annRec{\seq{\bind{x}{v}}}{\alpha},
         G' \cup \inStar{V}{\alpha}
      }
      \and
      \inferrule*[lab={\ruleName{$\evalS$-constr}}
      , right={$\Sigma(c) = |\seq{e}|$}]
      {
         \gamma, \seq{e}, V, G \evalS \seq{v}, G'
         \\
         \fresh{\alpha}{G'}
      }
      {
         \gamma, \exConstr{c}{\seq{e}}, V, G
         \evalS
         \annConstr{c}{\seq{v}}{\alpha},
         G' \cup \inStar{V}{\alpha}
      }
      \and
      \inferrule*[
         lab={\ruleName{$\evalS$-project}}
      ]
      {
         \gamma, e, V, G \evalS \annRec{\seq{\bind{x}{v}} \cons (\bind{y}{u})}{\alpha}, G'
      }
      {
         \gamma, \exRecProj{e}{y}, V, G
         \evalS
         u,
         G'
      }
      \and
      \inferrule*[ lab={\ruleName{$\evalS$-foreign-app}}
      , right={$\Phi(f) = |\seq{e}|$}]
      {
         \gamma, \seq{e}, V, G_1 \evalS \seq{v}, G_2
         \\
         \interpret{f}(\seq{v}, G_2) = (u, G_3)
      }
      {
         \gamma, \exForeignApp{f}{\seq{e}}, V, G_1
         \evalS
         u,
         G_3
      }
      \and
      \inferrule*[
         lab={\ruleName{$\evalS$-let}}
      ]
      {
         \gamma, e, V, G_1 \evalS v, G_2
         \\
         \gamma \cons (\bind{x}{v}), e', V, G_2 \evalS v', G_3
      }
      {
         \gamma, \exLet{x}{e}{e'}, V, G_1
         \evalS
         v',
         G_3
      }
      \and
      \inferrule*[
         lab={\ruleName{$\evalS$-let-rec}}
      ]
      {
         \gamma, \rho, V, G_1 \closeDefs  \gamma', G_2
         \\
         \gamma \concat \gamma', e, V, G_2 \evalS v, G_3
      }
      {
         \gamma, \exLetRec{\rho}{e}, V, G_1
         \evalS
         v,
         G_3
      }
      \and
      \inferrule*[
         lab={\ruleName{$\evalS$-app}},
         width=5.5in,
      ]
      {
         \gamma, e, V, G_1 \evalS \annClosure{\gamma_1}{\rho}{\sigma}{\alpha}, G_2
         \\
         \gamma_1, \rho, \set{\alpha}, G_2 \closeDefs \gamma_2, G_3
         \\
         \gamma, e', V, G_3 \evalS v', G_4
         \\
         v', \sigma \match \gamma_3, e^\twoPrime, V'
         \\
         \gamma_1 \concat \gamma_2 \concat \gamma_3, e^\twoPrime,  V' \cup \set{\alpha}, G_4 \evalS u, G_5
      }
      {
         \gamma, \exApp{e}{e'}, V, G_1
         \evalS
         u,
         G_5
      }
   \end{smathpar}}
   \\[-1cm]
   {\small \flushleft \shadebox{$\gamma, \seq{e}, V, G \evalS \seq{v}, G'$}%
   \begin{smathpar}
      \inferrule*[
         right={$n = \length{\seq{e}}$}
      ]
      {
         \gamma, e_i, V, G_i \evalS v_i, G_{i + 1}
         \\
         (\forall i \numleq n)
      }
      {
         \gamma, \seq{e}, V, G_1
         \evalS
         \seq{v},
         G_{n + 1}
      }
      \and
   \end{smathpar}}
   {\small \flushleft \shadebox{$\gamma, \rho, V, G \closeDefs \gamma', G'$}
   \begin{smathpar}
      \inferrule*[
         right={$n = \length{\seq{x}}$}
      ]
      {
         \gamma'(x_i) = \annClosure{\gamma}{\rho}{\rho(x_i)}{\alpha_i}
         \\
         \alpha_i \notin \dom{G_i}
         \\
         G_{i+1} = G_i \cup \inStar{V}{\alpha_i}
         \quad
         (\forall i \numleq n)
      }
      {
         \gamma, \rho, V, G_1
         \closeDefs
         \gamma',
         G_{n+1}
      }
   \end{smathpar}
   }
   \caption{Operational semantics with dependence graph}
   \label{fig:core:eval}
\end{figure}

\subsubsection{Pattern matching}
\label{sec:core:semantics:pattern-matching}

\figref{core:pattern-matching} defines the pattern-matching judgement $\seq{v}, \kappa \match \gamma, e, V$.
Rather than matching a single value $v$, the judgement matches a ``stack'' of values $\seq{v}$ against a
continuation $\kappa$, returning the selected branch $e$, an environment $\gamma$ providing bindings for the
free variables of $e$, and the set $V$ of addresses found in the matched portions of $\seq{v}$.

A continuation which is just an expression $e$ matches only the empty stack of values
(\ruleName{$\match$-done}), in which case pattern-matching is complete and $e$ is the selected branch. The
other rules require an \emph{eliminator} $\sigma$ as the continuation and a non-empty stack $v \cons \seq{v}$;
any relevant subvalues of $v$ are unpacked and pushed onto the tail $\seq{v}$ and then recursively matched
using the continuation $\kappa$ selected from $\sigma$. A variable eliminator $\elimVar{x}{\kappa}$ pops $v$
off the stack, using $\kappa$ to recurse (\ruleName{$\match$-var}); no part of $v$ is consumed so the
addresses $V$ consumed by the recursive match is returned unmodified. A record eliminator
$\elimRecord{\seq{y}}{\kappa}$ matches a record of the form $\annot{\exRec{\seq{\bind{x}{v}}}}{\alpha}$ as
long as the variables in $\seq{y}$ are also fields in $\seq{x}$, argumenting $V$ with the address $\alpha$
associated with the record (\ruleName{$\match$-record}); the premise $\exRec{\seq{\bind{y}{u}}} \subseteq
\exRec{\seq{\bind{x}{v}}}$ projects out the corresponding values of $\seq{u}$ from $\seq{v}$. Lastly, a
constructor eliminator $(\elimBind{c}{\kappa})$ matches any constructor value of the form
$\annot{\exConstr{c}{\seq{v}}}{\alpha}$, argumenting $V$ with the address $\alpha$ associated with the
constructor (\ruleName{$\match$-constr}).

\subsubsection{Evaluation}

\figref{core:eval} defines a big-step evaluation relation $\gamma, e, V, G \evalS v, G'$ stating that term
$e$, under an environment $\gamma$, vertex set $V$ and dependence graph $G$, evaluates to a value $v$ and
extended dependence graph $G'$. The vertex set $V$ records the (partial) input values consumed by the current
active function call providing the dynamic context in which $e$ is being evaluated; $V$ is initially empty and
changes whenever a function application is evaluated.

The evaluation rule for variables is fairly standard; the dependence graph is returned unmodified, because no
new addresses are allocated as a result of simply looking up a variable. The rule for record projections
$\exRecProj{e}{y}$ is similar: if $e$ evalutes to a record of the form $\annRec{\seq{\bind{x}{v}} \cons
(\bind{y}{u})}{\alpha}$, then $\exRecProj{e}{y}$ evalutes to the value $u$ of field $y$, discarding the
address $\alpha$ of the record.

Introduction rules, such as for integers, functions, records and constructors, follow the pattern of
assigning a fresh address for the (partial) value being constructed, and then extending $G$ with a set of
dependency edges from the vertices in $V$ to $\alpha$. For example, the integer rule evaluates an expression
$n$ to its value form $n_{\alpha}$; the $\alpha \notin \V(G)$ constraint ensures that $\alpha$ is fresh. The
dependence graph is then extended with $\inStar{V}{\alpha}$, indicating that $n_\alpha$ depended on all
matched partial inputs in $V$. Likewise, the rule for an anonymous function $\exFun{\sigma}$ constructs the
closure $\exClosure{\gamma}{\envEmpty}{\sigma}_\alpha$ with fresh address $\alpha$, capturing the current
environment $\gamma$ and using $\envEmpty$ as the set of mutually recursive definitions associated with the
function, and again establishing dependency edges from $V$ to $\alpha_i$.

Rules which involve recursively evaluating subterms thread the graph under construction through the evaluation
of the subterms. For example the auxiliary evaluation relation $\gamma, \seq{e}, V \evalS \seq{v}, G$ (bottom
of \figref{core:eval}), evalutes each $e_i$ in a sequence of terms $\seq{e}$ to a value $v_i$ and dependence
graph $G_{i + 1}$, which is used as the input graph for evaluating $e_{i + 1}$. We make use of this judgement
in the rules for records $\exRec{\seq{\bind{x}{e}}}$, constructors $c(\seq{e})$, and foreign applications
$f(\seq{e})$; in the last rule, $\hat{f}$ is the foreign implementation that evaluates an application of $f$
to a sequence of values $\seq{v}$ and dependence graph $G$ to a result $v$ and extended dependence graph $G'$
(\secref{core:foreign}).

The rules for $\exLet{x}{e}{e'}$ and application $\exApp{e}{e'}$ may involve mutual recursion, and rely on the
auxiliary relation $\gamma, \rho, V, G \closeDefs \gamma', G'$  defined at the bottom of \figref{core:eval}.
This judgement takes a set $\rho$ of recursive definitions, vertex set $V$ and dependence graph $G$, and
returns an environment $\gamma'$ of closures derived from $\rho$ and extended dependence graph $G$. Each
function definition $\rho(x_i) \in \rho$ generates a new closure $(\gamma, \rho, \rho(x_i))_{\alpha_i}$
capturing $\gamma$ and $\rho$, with fresh address ${\alpha_i}$, and extends the dependence graph with a set of
edges from $V$ to $\alpha$.

The evaluation rule for $\exLetRec{\rho}{e}$ is similar to the rule for regular let-bindings, except for using
$\closeDefs$ to build a set of closures $\gamma'$ which is used to extend $\gamma$. Finally, the rule for
$\exApp{e}{e'}$ is notable because this is where the active function context changes and the ambient $V$ is
discarded. We compute the closure $\annClosure{\gamma_1}{\rho}{\sigma}{\alpha}$ from $e$ and the argument $v'$
from $e'$, and then use $\sigma$ to match $v'$. If pattern-matching returns selected branch $e^\twoPrime$,
with vertex set $V'$ representing the consumed part of $v'$, then $e^\twoPrime$ is evaluated and becomes the
result of the application, with $V' \cup \set{\alpha}$ serving as set of (partial) inputs associated with the
new active function context.

\definecolor{myblue}{RGB}{172,216,235}
\begin{figure}
   \centering
   \begin{tikzpicture}
      [scale=0.9,
               arrow/.style={->,black},
               set name/.style={font=\color{myblue}\tiny\bfseries\sf},
               set/.style={thick, myblue},
               every node/.style={circle},
               font=\sf
      ]
      \begin{scope}[name prefix = input-]
         \node[fill=myblue, label={below:{$18.17$}}] (1) at (2, 0) {};
         \node[fill=myblue, label={below:{$22.13$}}] (2) at  (4, 0) {};
         \node[fill=myblue, label={below:{$37.14$}}] (3) at  (8, 0) {};
         \node[fill=myblue, label={below:{$61.27$}}] (4) at (10, 0) {};
      \end{scope}

      \begin{scope}[name prefix = len-]
         \node[fill=myblue, label=left:$2$] (div2) at (0.5, 1.5) {};
         \node[fill=myblue, label=right:$3$] (div3) at (7.0, 3.0) {};
      \end{scope}

      \begin{scope}[name prefix = add-]
         \node[fill=myblue, label={[label position = 0] $40.3$}] (1) at (2.5, 2) {+};
         \node[fill=myblue, label={[label position = 0] $40.3$}] (2) at (4.25, 2) {+};
         \node[fill=myblue, label={[label position = 180] $77.44$}] (3) at (5.0, 3) {+};
         \node[fill=myblue, label={[label position = 0] $59.27$}] (4) at (7.0, 2.0) {+};
         \node[fill=myblue, label={[label position = 0]$120.54$}] (5) at (8.5, 3) {+};
      \end{scope}

      \begin{scope}[name prefix = div-]
         \node[fill=myblue, label={$20.15$}] (1) at (2.5, 4) {/};
         \node[fill=myblue, label={$25.81$}] (2) at (5.0, 4) {/};
         \node[fill=myblue, label={$40.18$}] (3) at (8.5, 4) {/};
      \end{scope}

      \begin{scope}
         \draw[->] (input-1) -- (add-1);
         \draw[->] (input-2) -- (add-1);
         \draw[->] (add-1) -- (div-1);
         \draw[->] (len-div2) -- (div-1);
         \draw[->] (add-1) -- (div-1);
         \draw[->] (input-1) -- (add-2);
         \draw[->] (input-2) -- (add-2);
         \draw[->] (add-2) -- (add-3);
         \draw[->] (input-3) -- (add-3);
         \draw[->] (add-3) -- (div-2);
         \draw[->] (len-div3) -- (div-2);
         \draw[->] (input-2) -- (add-4);
         \draw[->] (input-3) -- (add-4);
         \draw[->] (input-4) -- (add-5);
         \draw[->] (add-4) -- (add-5);
         \draw[->] (add-5) -- (div-3);
         \draw[->] (len-div3) -- (div-3);
      \end{scope}
   \end{tikzpicture}
   \caption{Portion of dependence graph computed for moving averages example}
   \label{fig:core:mavg}
\end{figure}

\subsubsection{Example}

\figref{core:mavg} shows part of the dependence graph for the moving average example in
\figrefTwo{introduction:related-inputs}{example:moving-average:source}. The graph shows the calculation of the
first 3 points in the moving average plot from entries in the \emph{emissions} table. This is a simplified
version of the graph, since in practise they get very large; for example we omit closures, list cells and the
calculation of the divisors. Also note that the node labels here are merely illustrative: the graph only
stores value dependencies, and in particular the in-neighbours of a vertex are unordered.

\section{Cognacy Queries Over Dependence Graphs}
\label{sec:conjugate}

We now turn to cognacy and ``co-cognacy'' queries over dependence graphs, which are expressed in terms of the
$\demandsR$ and $\demandedByR$ operators introduced informally in \secref{overview}. Boolean algebras
(\defref{conjugate:Boolean-algebra}) are used to represent selections; we then define $\demandsR$ and
$\demandedByR$ and their De Morgan duals, first for an arbitrary relation (\secref{conjugate:image-preimage})
and then for the \emph{IO relation} of a dependence graph (\secref{conjugate:dependence-graph-functions}).
Then we give procedures for computing $\demandsR$ and $\demandedByR$ and their duals over a given dependence
graph, via an intermediate graph slice (\secref{conjugate:dependence-graph-algos}).

\begin{definition}[Boolean algebra]
\label{def:conjugate:Boolean-algebra}
A \emph{Boolean algebra} (or \emph{Boolean lattice}) $A$ is a 6-tuple $(A, \meet, \join, \bot, \top, \neg)$
with carrier $A$, distinguished elements $\bot$ (bottom) and $\top$ (top), commutative operations $\meet$
(meet) and $\join$ (join) with $\top$ and $\bot$ as respective units, and unary operation $\neg$ (negate)
satisfying $x \join \neg x = \top$ and $x \meet \neg x = \bot$. The operations $\meet$ and $\join$ distribute
over each other and induce a partial order $\leq$ over $X$.
\end{definition}

An input (or output) selection, where $X$ is the set of all addresses that occur in the input (or output) of
the program, is represented by the \emph{powerset} Boolean algebra $\powerset{X}$ whose carrier is the set of
subsets $X' \subseteq X$ ordered by inclusion $\subseteq$. In this case bottom is the empty set $\emptyset$,
top is the whole set $X$, meet and join are given by $\cap$ and $\cup$, and negation by relative complement
$\setminus$, so that $\neg X' = X \setminus X'$.

\subsection{$\demandsR$ and $\demandedByR$ and their duals}
\label{sec:conjugate:image-preimage}

The $\demandsR$ and $\demandedByR$ operators sketched earlier are simply the image and preimage functions for
a particular relation $R$.

\begin{definition}[Image and Preimage Functions for a Relation]
   For a relation $R \subseteq X \times Y$, define $\demandedByR_R: \powerset{X} \to \powerset{Y}$ and
   $\demandsR_R : \powerset{Y} \to \powerset{X}$ as:
   \begin{enumerate}
      \item $\demandedByR_{R}(X') \eqdef \set{y \in Y \mid \exists x \in X'.(x,y) \in R}$ \hfill (image)
      \item $\demandsR_{R}(Y') \eqdef \set{x \in X \mid \exists y \in Y'. (x,y) \in R}$ \hfill (preimage)
   \end{enumerate}
\end{definition}

\noindent Trivially $\demandedByR_{R} = \demandsR_{R^{-1}}$. The image and preimage functions form a
\textit{conjugate pair}, in the sense of \citet{jonsson51}:

\begin{definition}[Conjugate Functions]
   For Boolean algebras $A$, $B$, functions $f: A \to B$ and $g: B \to A$ form a \textit{conjugate pair} iff:
   $$
   f(x) \meet y = \bot \iff x \meet g(y) = \bot
   $$
\end{definition}


\begin{lemma}
\label{lem:im-preim-conjugate}
$\demandedByR_{R}$ and $\demandsR_{R}$ are conjugate.
\end{lemma}

This should be intuitive enough: for any subsets $X' \subseteq X$ and $Y' \subseteq Y$, if the elements ``on
the right'' to which $X'$ is related are disjoint from $Y'$, then there are no edges in $R$ from $X'$ to $Y'$;
and in virtue of that, the elements ``on the left'' to which $Y'$ is related must also be disjoint from $X'$.

Conjugate functions are related to another class of near-reciprocals between Boolean algebras, namely
\emph{Galois connections}; in fact every conjugate pair induces a Galois connection~\cite{menni14}. This
relates the present setting to previous work on program slicing with Galois
connections~\cite{perera12a,perera16d,ricciotti17}.

\begin{definition}[Galois connection]
   \label{def:conjugate:galois-connection}
   Suppose $A, B$ are partial orders. Then monotone functions $f: A \to B$ and $g: B \to A$ form a \textit{Galois
   connection} iff
   \[ f(x) \leq y \iff x \leq g(y) \]
\end{definition}

\begin{proposition}
   \label{prop:conj-galois-correspond}
   Suppose functions $f: A \to B$ and $g: B \to A$ between Boolean algebras. The following statements are
   equivalent:
   \begin{enumerate}
      \item $f$ and $g$ form a conjugate pair
      \item $f$ and $\dual{g}$ form a Galois connection
   \end{enumerate}
\end{proposition}

It is also useful (both for performance reasons and as a user feature) to consider the \emph{De Morgan duals}
of $\demandsR$ and $\demandedByR$, which we write as $\sufficesR$ and $\preimageDualR$; these also form a
conjugate pair.

\begin{definition}[De Morgan Dual]
   \label{def:conjugate:de-morgan-dual}
   Suppose a function $f: A \to B$ between Boolean algebras with $\neg_A$ and $\neg_B$ the negation operators
   of $A$ and $B$. Define the \textit{De Morgan Dual} $\dual{f}: A \to B$ of $f$ as:
   \[ \dual{f} \eqdef \neg_B \after f \after \neg_A \]
\end{definition}

\begin{definition}[Dual Image and Preimage Functions for a Relation]
   For relations $R \subseteq X \times Y$, define $\sufficesR_R: \powerset{X} \to \powerset{Y}$ and
   $\preimageDualR_R: \powerset{Y} \to \powerset{X}$ as:
   \begin{enumerate}
      \item $\sufficesR_R(X') \eqdef \set{y \in Y \mid \nexists x \in X \setminus X'.(x, y) \in R}$ \hfill
      (dual image)
      \item $\preimageDualR_R(Y') \eqdef \set{x \in X \mid \nexists y \in Y \setminus Y'.(x, y) \in R}$ \hfill
      (dual preimage)
   \end{enumerate}
\end{definition}

Bearing in mind that $\neg$ for $\powerset{X}$ is just relative complement $X \setminus \param$, it is easy to
show that these are indeed the intended De Morgan duals.

\begin{lemma}[Duality of image and preimage functions]
   \label{lem:conjugate:image-preimage-duality}
   \item
   \begin{enumerate}
      \item $\dual{\demandedByR_R} = \sufficesR_R$
      \item $\dual{\demandsR_R} = \preimageDualR_R$
   \end{enumerate}
\end{lemma}

\noindent If $\demandedByR_R(X')$ picks out the outputs that the elements of $X'$ are \emph{necessary} for,
$\sufficesR_R(X')$ picks out the outputs that the elements of $X$ are \emph{sufficient} for.
$\preimageDualR_R(Y')$ picks out the inputs that are needed \emph{only} by elements of $Y'$.


\subsection{$\demandsR_G$ and $\demandedByR_G$ for Dependence Graphs}
\label{sec:conjugate:dependence-graph-functions}


Reachability in $G$ induces a relation just between \emph{sources} and \emph{sinks}, which we call the
\emph{IO} relation of $G$. We now extend $\demandsR_{R}$ and $\demandedByR_{R}$ and their duals to a
dependence graph $G$, via its IO relation.

\begin{definition}[Sources and sinks]
For a graph $G = (V,E)$, write $\sources{G}$ for the \emph{sources} of $G$ (i.e.~those vertices with no
in-edges) and $\sinks{G}$ for the \emph{sinks} of $G$ (those with no out-edges).
\end{definition}

\begin{definition}[Reachability relation]
Define the \emph{reachability relation} for a graph $G=(V,E)$ to be the reflexive transitive closure of $E$.
\end{definition}

\begin{definition}[IO relation]
For any dependence graph $G$ with reachability relation $R$, define the \emph{IO} relation of $G$ to be $R
\cap (\sources{G} \times \sinks{G})$.
\end{definition}

The IO relation specifies how specific inputs (sources) are demanded by specific outputs (sinks); thus we
interpret $(x, y) \in R$ as ``$x$ is \textit{demanded by} $y$''. Clearly the IO relation of $G$ is the
converse of the IO relation of $\opGraph{G}$. The set of inputs that some outputs $Y$ \emph{demand},
$\demandsR_{G}(Y)$, is simply the subset of the inputs of $G$ reachable from $Y$ (by traversing the graph in
its opposite direction). Conversely, the set of outputs \emph{demanded by} some inputs $X$,
$\demandedByR_{G}(X)$, is simply the subset of the outputs of $G$ reachable from $X$.


\begin{definition}[$\demandsR_G$ and $\demandedByR_G$ for a dependence graph]
   For a graph $G$ with IO relation $R$, define:
   \begin{enumerate}
      \item $\demandsR_G \eqdef \demandsR_{R}: \powerset{\sinks{G}} \to \powerset{\sources{G}}$ \hfill
      (demands)
      \item $\demandedByR_G \eqdef \demandedByR_{R}: \powerset{\sources{G}} \to \powerset{\sinks{G}}$ \hfill
      (demanded by)
   \end{enumerate}
\end{definition}

\subsection{Computing $\demandsR_G$, $\demandedByR_G$, $\sufficesR_{G}$ and $\preimageDualR_{G}$ for
Dependence Graphs}
\label{sec:conjugate:dependence-graph-algos}

We now show how to compute $\demandsR_{G}$, $\demandedByR_{G}$, $\sufficesR_{G}$ and $\preimageDualR_{G}$ for
a dependence graph $G$, via an intermediate graph slice which we then restrict to its IO relation. First some
graph notation:

\begin{definition}[In-edges and out-edges]
For a graph $G=(V,E)$ and vertex $\alpha \in V$, write $\inE{G}(\alpha)$ for the in-edges of $\alpha$ in $G$
and $\outE{G}(\alpha)$ for its out-edges.
\end{definition}

\begin{definition}[Opposite graph]
For a graph $G$ define the \emph{opposite} graph $\opGraph{G} \eqdef (V,E^{-1})$ where $\param^{-1}$ denotes
relational converse.
\end{definition}

With the trace-based approaches mentioned in \secref{introduction}, one can use a given algorithm to implement
its De Morgan dual; for example given a procedure for $\demandedByR_{G}$ we can compute $\sufficesR_{G}$ by
pre- and post-composing with negation~\cite{perera22}. In the dependence graph setting we can also use a given
algorithm to compute its conjugate: for example given a procedure for $\demandedByR_{G}$, we can compute
$\demandsR_{G}$ simply as $\demandedByR_{\opGraph{G}}$.

For efficiency, however, we give direct procedures both for $\demandedByR_{G}$
(\secref{conjugate:algorithm:demandedBy}) and $\sufficesR_{G}$ (\secref{conjugate:algorithm:suffices}), which
can then be used to derive implementations of the other operators (\figref{conjugate:derivate-algorithms}
below). Each procedure factors through an auxiliary operation that computes a ``slice'' of the original graph,
i.e.~a contiguous subgraph. While it is technically possible to compute the desired image/preimage of the IO
relation without creating this intermediate graph, we anticipate use cases which will make use of the graph
slice; these are discussed in more detail in \secref{conclusion}. Our implementation makes it easy to flip
between $G$ and $\opGraph{G}$ so we freely make use of both in the algorithms to access out-neighbours and
in-neighbours.

\subsubsection{Direct algorithm for $\demandedByR_{G}$ (\emph{demanded by})}
\label{sec:conjugate:algorithm:demandedBy}

\begin{definition}[$\demandedByAlg{G}$]
\label{def:conjugate:algorithm:demandedBy}
\figref{conjugate:algorithms:demands} defines the family of relations $\demandedByAlg{G}$ for dependence graph $G$.
\end{definition}

For a dependence graph $G$ and inputs $X \subseteq \sources{G}$, the judgement $X\; \smash{\demandedByAlg{G}}
Y$ says that $X$ is demanded by $Y \subseteq \sinks{G}$. The algorithm defers to an auxiliary operation
$\demandedByV{}$, which takes a (partial) graph slice $H \subseteq G$, initially containing only the original
input vertices $X$ and no edges, and which proceeds as follows. If there is a sink of $H$ that still has
outgoing edges in $G$, then the targets of those edges are reachable from the vertices of $H$, and so we must
add them to $H$ and recurse (\ruleName{extend}). If $\sinks{H} \subseteq \sinks{G}$, then there are no
unexplored nodes in $G$ that are reachable from $H$, and so we terminate with the current state of $H$
(\ruleName{done}). We note that it is enough to consider the sinks of $H$, since we move every outgoing edge
of a vertex from $G$ to $H$ all at once. When $\demandedByV{}$ is done, $\demandedByAlg{G}$ returns the sinks
$T(H)$ from the final graph slice $H$, representing the outputs that the original inputs are demanded by.

\begin{proposition}[$\demandedByAlg{G}$ Computes $\demandedByR_{G}$]
   \label{prop:demandedByAlg-computes-demandedBy}
   For any dependence graph $G$, any $X \subseteq \sources{G}$ and any $Y \subseteq \sinks{G}$ we have
   \[X \demandedByAlg{G} Y \iff \demandedByR_G(X) = Y\]
\end{proposition}
\begin{proof}
   See \apdxsecref{apdx:demandedBy}.
\end{proof}

\begin{figure}
   \begin{subfigure}{\textwidth}
      {\small \flushleft \shadebox{$X \demandedByAlg{G} Y$}
         \begin{smathpar}
            \inferrule*[
              right={$X \subseteq \sources{G}$}
            ]
            {
               (X, \emptyset), G \demandedByV{} H
            }
            {
               X \demandedByAlg{G} \sinks{H}
            }
            \end{smathpar}
         }
         \vspace{-5mm}
         {\small \flushleft \shadebox{$H, G \demandedByV{} H'$}
         \begin{smathpar}
         \inferrule*[
           lab={\ruleName{done}},
           right={$\sinks{H} \subseteq \sinks{G}$}
         ]
         {
            \strut
         }
         {
            H, G \demandedByV{} H
         }
         \and
         \inferrule*[
            lab={\ruleName{extend}},
            right={$\alpha \in \sinks{H} \wedge E = \outE{G}(\alpha) \neq \emptyset$}
         ]
         {
            H \cup E, G \edgeMinus E
            \demandedByV{}
            H'
         }
         {
            H, G \demandedByV{} H'
         }
         \end{smathpar}}
         \caption{$\demandedByAlg{G}$ algorithm}
         \label{fig:conjugate:algorithms:demands}
   \end{subfigure}
   \\[1.5em]
   \begin{subfigure}{\textwidth}
   {\small \flushleft \shadebox{$X \sufficesAlg{G} U$}}
   \begin{smathpar}
      \inferrule*[
         right=$X \subseteq \sources{G}$
      ]
      {
         (X, \emptyset),
         (\emptyset, \emptyset),
         G
         \sufficesE{}
         H
      }
      {
         X \sufficesAlg{G} \V(H) \cap \sinks{G}
      }
   \end{smathpar}
   \vspace{-5mm}
   {\small \flushleft \shadebox{$H, P, G \sufficesE{} H'$}}
   \begin{smathpar}
      \inferrule*[
         lab={\ruleName{done}},
         right={$\sources{G} \cap \V(P) \subseteq \sources{P} \wedge \V(H) \subseteq \sinks{G}$}
      ]
      {
         \strut
      }
      {
         H, P, G \sufficesE{} H
      }
      \and
      \inferrule*[
         lab={\ruleName{pending}},
         right={$\alpha \in \V(H) \wedge E = \outE{G}(\alpha) \neq \emptyset$}
      ]
      {
         H, P \cup E, G \edgeMinus E
         \sufficesE{}
         H'
      }
      {
         H, P, G
         \sufficesE{}
         H'
      }
      \and
      \inferrule*[
         lab={\ruleName{extend}},
         right={$\alpha \in \sources{G} \wedge E = \inE{P}(\alpha) \neq \emptyset$}
      ]
      {
         H \cup E,
         P \edgeMinus E,
         G
         \sufficesE{}
         H'
      }
      {
         H, P, G
         \sufficesE{}
         H'
      }
   \end{smathpar}
   \caption{$\sufficesAlg{G}$ algorithm}
   \label{fig:conjugate:slicing:forward}
   \end{subfigure}
   \caption{Dual analyses over a dependence graph $G$}
   \label{fig:conjugate:algorithms}
\end{figure}

\subsubsection{Direct algorithm for $\sufficesR_{G}$ (\emph{suffices})}
\label{sec:conjugate:algorithm:suffices}

\begin{definition}[$\sufficesAlg{G}$]
\label{def:conjugate:algorithm:suffices}
\figref{conjugate:slicing:forward} defines the family of relations $\sufficesAlg{G}$ for dependence graph $G$.
\end{definition}

Like $\demandsAlg{G}$, the algorithm $\sufficesAlg{G}$ also delegates to an auxiliary operation $\sufficesE{}$
which builds a slice of $G$. The operation $\sufficesE{}$ takes the graph slice $H$ under construction, a
``pending'' subgraph $P$ of nodes for which we have discovered partial information, and the remaining
unexplored graph $G$, and returns the final graph slice $H'$. Whenever we have a vertex in $H$ which still has
outgoing edges in $G$, we add those edges and their endpoints to the pending graph $P$ (\ruleName{pending}).
If a vertex $\alpha$ has no more incoming edges in $G$, we can then move $\alpha$ and its incoming edges from
$P$ to $H$ (\ruleName{extend}), as this only happens when we have already moved every ancestor of $\alpha$
into $H$. If neither scenario is the case, then we terminate with the (potentially incomplete) graph
$H$~(\ruleName{done}). Once $\sufficesE{}$ has terminated with $H$, $\sufficesAlg{G}$ returns just the
vertices in $H$ that are also sinks in $G$. If there are no such vertices, it simply means that the input data
$X$ was insufficient to compute any of the original program's outputs.

We provide an example of a run of the algorithm in \figref{suffices-algo-diagram}. Here,
the green vertices and thick edges are in $H$, the orange vertices and dashed arrows are in $P$, and the
blue vertices and normal arrows are in $G$. In the figure, steps 2, 3, 5 and 6 correspond to applications
of the rule \ruleName{pending}, and steps 4 and 7 correspond to applications of the rule \ruleName{extend}.
After step 7, the run is complete.

\definecolor{myblue}{RGB}{172,216,235}
\definecolor{burnt}{RGB}{204,85,0}
\definecolor{forest}{RGB}{0,144,21}

\begin{figure}
   \begin{subfigure}{0.24\textwidth} 
      \begin{tikzpicture}[scale=0.6,
               arrow/.style={->,black},
               set name/.style={font=\color{myblue}\tiny\bfseries\sf},
               set/.style={thick, myblue},
               every node/.style={scale=0.6,circle},
               font=\sf
            ]
         \node[fill=forest] (x1) at (1.0, 3.0) {};
         \node[fill=forest] (x2) at (1.0, 2.0) {};
         \node[fill=forest] (x3) at (1.0, 1.0) {};
         \node[fill=myblue] (x4) at (1.0, 0.0) {};
         \node[fill=myblue] (x5) at (3.0, 2.5) {};
         \node[fill=myblue] (x6) at (5.0, 2.0) {};
         \node[fill=myblue] (x7) at (5.0, 0.75) {};
         \node[circle, draw = black](step) at (4.25, 3.0) {1};
         \begin{scope}[arrow]
            \draw (x1.east) -- (x5);
            \draw (x2.east) -- (x5);
            \draw (x3.east) -- (x6);
            \draw (x3.east) -- (x7);
            \draw (x4.east) -- (x7);
            \draw (x5.east) -- (x6);
         \end{scope}
      \end{tikzpicture}
   \end{subfigure}
   \begin{subfigure}{0.24\textwidth} 
      \begin{tikzpicture}[scale=0.6,
               arrow/.style={->,black},
               set name/.style={font=\color{myblue}\tiny\bfseries\sf},
               set/.style={thick, myblue},
               every node/.style={scale=0.6,circle},
               font=\sf
            ]
         \node[fill=forest] (x1) at (1.0, 3.0) {};
         \node[fill=forest] (x2) at (1.0, 2.0) {};
         \node[fill=forest] (x3) at (1.0, 1.0) {};
         \node[fill=myblue] (x4) at (1.0, 0.0) {};
         \node[fill=burnt] (x5) at (3.0, 2.5) {};
         \node[fill=myblue] (x6) at (5.0, 2.0) {};
         \node[fill=myblue] (x7) at (5.0, 0.75) {};
         \node[circle, draw = black](step) at (4.25, 3.0) {2};
         \begin{scope}[arrow]
            \draw[dashed] (x1.east) -- (x5);
            \draw (x2.east) -- (x5);
            \draw (x3.east) -- (x6);
            \draw (x3.east) -- (x7);
            \draw (x4.east) -- (x7);
            \draw (x5.east) -- (x6);
         \end{scope}
      \end{tikzpicture}
   \end{subfigure}
   \begin{subfigure}{0.24\textwidth} 
      \begin{tikzpicture}[scale=0.6,
               arrow/.style={->,black},
               set name/.style={font=\color{myblue}\tiny\bfseries\sf},
               set/.style={thick, myblue},
               every node/.style={scale=0.6,circle},
               font=\sf
            ]
         \node[fill=forest] (x1) at (1.0, 3.0) {};
         \node[fill=forest] (x2) at (1.0, 2.0) {};
         \node[fill=forest] (x3) at (1.0, 1.0) {};
         \node[fill=myblue] (x4) at (1.0, 0.0) {};
         \node[fill=burnt] (x5) at (3.0, 2.5) {};
         \node[fill=myblue] (x6) at (5.0, 2.0) {};
         \node[fill=myblue] (x7) at (5.0, 0.75) {};
         \node[circle, draw = black](step) at (4.25, 3.0) {3};
         \begin{scope}[arrow]
            \draw[dashed] (x1.east) -- (x5);
            \draw[dashed] (x2.east) -- (x5);
            \draw (x3.east) -- (x6);
            \draw (x3.east) -- (x7);
            \draw (x4.east) -- (x7);
            \draw (x5.east) -- (x6);
         \end{scope}
      \end{tikzpicture}
   \end{subfigure}
   \begin{subfigure}{0.24\textwidth} 
      \begin{tikzpicture}[scale=0.6,
               arrow/.style={->,black},
               set name/.style={font=\color{myblue}\tiny\bfseries\sf},
               set/.style={thick, myblue},
               every node/.style={scale=0.6,circle},
               font=\sf
            ]
         \node[fill=forest] (x1) at (1.0, 3.0) {};
         \node[fill=forest] (x2) at (1.0, 2.0) {};
         \node[fill=forest] (x3) at (1.0, 1.0) {};
         \node[fill=myblue] (x4) at (1.0, 0.0) {};
         \node[fill=forest] (x5) at (3.0, 2.5) {};
         \node[fill=myblue] (x6) at (5.0, 2.0) {};
         \node[fill=myblue] (x7) at (5.0, 0.75) {};
         \node[circle, draw = black](step) at (4.25, 3.0) {4};
         \begin{scope}[arrow]
            \draw[thick] (x1.east) -- (x5);
            \draw[thick] (x2.east) -- (x5);
            \draw (x3.east) -- (x6);
            \draw (x3.east) -- (x7);
            \draw (x4.east) -- (x7);
            \draw (x5.east) -- (x6);
         \end{scope}
      \end{tikzpicture}
   \end{subfigure}
   \begin{subfigure}{0.24\textwidth} 
      \begin{tikzpicture}[scale=0.6,
               arrow/.style={->,black},
               set name/.style={font=\color{myblue}\tiny\bfseries\sf},
               set/.style={thick, myblue},
               every node/.style={scale=0.6,circle},
               font=\sf
            ]
         \node[fill=forest] (x1) at (1.0, 3.0) {};
         \node[fill=forest] (x2) at (1.0, 2.0) {};
         \node[fill=forest] (x3) at (1.0, 1.0) {};
         \node[fill=myblue] (x4) at (1.0, 0.0) {};
         \node[fill=forest] (x5) at (3.0, 2.5) {};
         \node[fill=myblue] (x6) at (5.0, 2.0) {};
         \node[fill=myblue] (x7) at (5.0, 0.75) {};
         \node[circle, draw = black](step) at (4.25, 3.0) {5};
         \begin{scope}[arrow]
            \draw[thick] (x1.east) -- (x5);
            \draw[thick] (x2.east) -- (x5);
            \draw (x3.east) -- (x6);
            \draw (x3.east) -- (x7);
            \draw (x4.east) -- (x7);
            \draw[dashed] (x5.east) -- (x6);
         \end{scope}
      \end{tikzpicture}
   \end{subfigure}
   \begin{subfigure}{0.24\textwidth} 
      \begin{tikzpicture}[scale=0.6,
               arrow/.style={->,black},
               set name/.style={font=\color{myblue}\tiny\bfseries\sf},
               set/.style={thick, myblue},
               every node/.style={scale=0.6,circle},
               font=\sf
            ]
         \node[fill=forest] (x1) at (1.0, 3.0) {};
         \node[fill=forest] (x2) at (1.0, 2.0) {};
         \node[fill=forest] (x3) at (1.0, 1.0) {};
         \node[fill=myblue] (x4) at (1.0, 0.0) {};
         \node[fill=forest] (x5) at (3.0, 2.5) {};
         \node[fill=myblue] (x6) at (5.0, 2.0) {};
         \node[fill=myblue] (x7) at (5.0, 0.75) {};
         \node[circle, draw = black](step) at (4.25, 3.0) {6};
         \begin{scope}[arrow]
            \draw[thick] (x1.east) -- (x5);
            \draw[thick] (x2.east) -- (x5);
            \draw[dashed] (x3.east) -- (x6);
            \draw[dashed] (x3.east) -- (x7);
            \draw (x4.east) -- (x7);
            \draw[dashed] (x5.east) -- (x6);
         \end{scope}
      \end{tikzpicture}
   \end{subfigure}
   \begin{subfigure}{0.24\textwidth} 
      \begin{tikzpicture}[scale=0.6,
               arrow/.style={->,black},
               set name/.style={font=\color{myblue}\tiny\bfseries\sf},
               set/.style={thick, myblue},
               every node/.style={scale=0.6,circle},
               font=\sf
            ]
         \node[fill=forest] (x1) at (1.0, 3.0) {};
         \node[fill=forest] (x2) at (1.0, 2.0) {};
         \node[fill=forest] (x3) at (1.0, 1.0) {};
         \node[fill=myblue] (x4) at (1.0, 0.0) {};
         \node[fill=forest] (x5) at (3.0, 2.5) {};
         \node[fill=myblue] (x6) at (5.0, 2.0) {};
         \node[fill=forest] (x7) at (5.0, 0.75) {};
         \node[circle, draw = black](step) at (4.25, 3.0) {7};
         \begin{scope}[arrow]
            \draw[thick] (x1.east) -- (x5);
            \draw[thick] (x2.east) -- (x5);
            \draw[thick] (x3.east) -- (x6);
            \draw[dashed] (x3.east) -- (x7);
            \draw (x4.east) -- (x7);
            \draw[thick] (x5.east) -- (x6);
         \end{scope}
      \end{tikzpicture}
   \end{subfigure}
   \caption{Run of $\sufficesAlg{}$ with $G$ and $H$ superimposed on $G_0$}
   \label{fig:suffices-algo-diagram}
\end{figure}

\begin{proposition}[$\sufficesAlg{G}$ Computes $\sufficesR_{G}$]
   \label{prop:sufficesAlg-computes-suffices}
   For any dependence graph $G$, any $X \subseteq \sources{G}$ and any $Y \subseteq \sinks{G}$ we have:
   \[X \sufficesAlg{G} Y \iff \sufficesR_{g}(X) = Y\]
\end{proposition}
\begin{proof}
   See \apdxsecref{apdx:suffices-alg}.
\end{proof}

\subsubsection{Derived algorithms}

\proprefTwo{demandedByAlg-computes-demandedBy}{sufficesAlg-computes-suffices} justify treating
$\demandedByAlg{G}$ and $\sufficesAlg{G}$ as functions. Using \lemref{conjugate:image-preimage-duality} we can
define a direct algorithm for $\demandsR_{G}$ from $\demandedByAlg{G}$ by simply flipping the graph, and the
same can be done to acquire an algorithm for $\preimageDualR_{G}$ from $\sufficesAlg{G}$. Using
\lemref{conjugate:image-preimage-duality} and \defref{conjugate:de-morgan-dual} we can also define alternative
algorithms for all 4 operators using the De Morgan dual construction. All algorithms are summarised in
\figref{conjugate:derivate-algorithms}.

\begin{figure}[H]
\centering
\begin{tabular}{|c|l|l|c|c|}
\emph{Abstract operator} & \emph{Direct algorithm} & \emph{Alternative algorithm} \\
$\demandedByR_{G}$ & $\demandedByAlg{G}$ & $\dual{\sufficesAlg{G}}$ \\
$\demandsR_{G}$ & $\demandedByAlg{\opGraph{G}}$ & $\dual{\sufficesAlg{\opGraph{G}}}$ \\
$\sufficesR_{G}$ & $\sufficesAlg{G}$ & $\dual{\demandedByAlg{G}}$ \\
$\preimageDualR_{G}$ & $\sufficesAlg{\opGraph{G}}$ & $\dual{\demandedByAlg{\opGraph{G}}}$
\end{tabular}
\caption{Direct and alternative (De Morgan Dual) algorithms for all 4 operators}
\label{fig:conjugate:derivate-algorithms}
\end{figure}

\noindent In \secref{evaluation} we contrast the performance of the direct and De Morgan dual implementations
of $\demandedByR_{G}$, which is an essential component of the linked inputs and linked outputs operators
introduced in \secref{introduction}.

\subsection{$\relInputR_G$ and $\relOutputR_G$ for Dependence Graphs}

Now we can provide a formal account of the notions of linked inputs and linked outputs.

\begin{definition}[Linked inputs and outputs]
   For a dependence graph $G$ define:
   \begin{enumerate}
      \item $\relInputR_{G} \eqdef \demandsR_{G} \after \demandedByR_{G}$\hfill (linked inputs)
      \item $\relOutputR_{G} \eqdef \demandedByR_{G} \after \demandsR_{G}$\hfill (linked outputs)
   \end{enumerate}
\end{definition}

Intuitively, linked outputs and linked inputs are relations of \emph{cognacy} (common ancestry) in $G$ and
$\opGraph{G}$ respectively. For a set of inputs $X \subseteq \sources{G}$, linked inputs asks ``what other
inputs are demanded by the outputs which demand X?''; it first finds all outputs $\demandedByR_{G}(X)$ that
our inputs are demanded by, and then computes the inputs that those outputs demand,
i.e.~$\demandsR_{G}(\demandedByR_{G}(X'))$. Conversely for a set of outputs $Y \subseteq \sinks{G}$, linked outputs asks
``what other outputs demand the inputs that Y demands''; it first finds all inputs $\demandsR_{G}(Y')$ that
our outputs demand, and then computes the outputs that those inputs $\demandsR_{G}(Y)$ demand,
i.e.~$\demandedByR_{G}(\demandsR_{G}(Y))$.

In general, neither $\relInputR_{G}$ nor $\relOutputR_{G}$ is inflationary; unputs which are not used anywhere
are not preserved by the $\relInputR_{G}$ round-trip, and outputs which don't require any inputs are not
preserved by the $\relOutputR_{G}$ round-trip. Thus in general we have neither $X' \subseteq
\relInputR_{G}(X')$ nor $X' \subseteq \relOutputR_{G}(X')$. However if we constrain things so that there are
no inputs which are unused, then $\demandsR_G$ will preserve $\bot$, i.e.~$\demandsR_G(\emptyset) =
\emptyset$, and $\relOutputR_{G}$ will preserve the original output selection. Similarly, if we constrain
things so that all outputs demand some input, then $\demandedByR_G$ will preserve $\bot$, and $\relInputR_{G}$
will preserve the original input selection:

\begin{lemma}
\item
\begin{enumerate}
   \item If $\demandedByR_G(\emptyset) = \emptyset$, then for all $X$ we have $\relInputR_{G}(X) \geq X$.
   \item If $\demandsR_G(\emptyset) = \emptyset$, then for all $Y$ we have $\relOutputR_{G}(Y) \geq Y$.
\end{enumerate}
\end{lemma}

As an easy corollary, if $\demandsR_G$ and $\demandedByR_G$ both preserve $\emptyset$ then in fact they form
an antitone Galois connection. In practical terms, it is possible to use this to produce a filtered view, only
showing inputs that are actually used in a program.

%

\section{Evaluation}
\label{sec:evaluation}

We now compare our dependence graph approach with the main alternative style of implementation, namely a
bidirectional interpreter with two components: a forward evaluator which performs a forwards analysis and
produces an execution trace, and a backward evaluator which consumes the execution trace and performs a
backwards analysis~\cite{ricciotti17,psallidas18smoke,perera22}.

For a language implementor, the main benefit of our graph approach is that it removes the need to implement a
bidirectional interpreter. Otherwise, changing the language requires modifying each direction of the
interpreter in a manner that maintains bidirectionality --- a considerable effort. In our approach the graph
algorithms for computing $\demandsR_{G}$ and $\demandedByR_{G}$ (\figref{conjugate:algorithms}) are
language-agnostic and only have to be defined once.

For a user, our hypothesis is that our approach has better performance than trace-based techniques. To test
this, we compare two implementations of the dependency-tracking runtime of \OurLang, one based on the
dependence graph design described in \secrefTwo{core}{conjugate}, and one based on execution traces and a
bidirectional interpreter. The other system components were shared by the two implementations, including the
parser, desugaring layer, visualisation layer and libraries. We wrote various benchmark programs taking
examples from data analysis and data visualisation (\secref{evaluation:choice-of-examples} below).

In an interactive system, performance plays a critical role in a user's experience. We therefore measured the
following:

\begin{itemize}
   \item Q1: Overhead of building a dependence graph versus building a trace;
   \item Q2: Performance of computing \textit{demands} over a graph vs. a trace;
   \item Q3: Performance of various implementations of \textit{demanded by}: trace-based, graph-based using
   $\demandsAlg{}$, and graph-based using the dual of $\sufficesAlg{}$.
\end{itemize}

\noindent We assessed performance according to the guidelines proposed by \citet{nielsen93}:
\begin{itemize}
   \item 100ms: limit for an interaction feeling \emph{instantaneous};
   \item 1000ms: limit for a user's train of thought to remain uninterrupted, although
   they may perceive a delay;
   \item 10000ms: limit to keep a user's attention on the system.
\end{itemize}

\noindent Ideally, interactive queries would stay as close as possible to the \emph{instantaneous} category
(Q2 and Q3), and the time taken to evaluate programs (Q1) would stay within the limit for keeping a user's
attention. Expriments were timed on an Intel Core i7-10850H 2.70GHz with 16Gb of RAM, using Google Chrome
121.0.6261.69 and JavaScript runtime V8 12.2.281.16. We used PureScript as the host language.

\subsubsection{Choice of examples}
\label{sec:evaluation:choice-of-examples}

We tested our system on a range of data analysis and visualisation benchmarks implemented in \OurLang. First,
we considered matrix convolution. Given that different kernels give rise to subtly different dependency
structures, we implemented three different kernels (\textbf{edge-detect}, \textbf{emboss}, \textbf{gaussian},
31 lines of code (LOC) each). Second, we developed a full-featured graphics library based on SVG that
demonstrates the performance on bespoke visualisation code (\textbf{grouped-bar-chart}, 140 LOC,
\textbf{line-chart}, 143 LOC, \textbf{stacked-bar-chart}, 136 LOC). Finally, we tested two examples that use a
D3.js front end for visualisation (\textbf{stacked-bar-scatter-plot}, 26 LOC, and
\textbf{bar-chart-line-chart}, 38 LOC), the same front end used in
\figrefTwo{introduction:related-inputs}{introduction:scatterplot}. The code for our benchmarks is given
in~\apdxsecref{apdx:evaluation-src-code}.

In the experiments, each benchmark is run 10 times, and we report the mean runtime and standard deviation (in
parentheses, coloured grey). Runtimes are reported in milliseconds, to 1 decimal place. When we refer to
\textit{speedup} or \textit{slowdown} in reference to a pair of implementations, we mean the ratio of average
runtimes for that benchmark.

\subsubsection{Q1: Overhead of graph construction vs. trace construction}

\tableref{eval-SpdUps} compares the average time taken to evaluate a program in the graph approach
(\textbf{G-Eval}) with the time taken in the trace-based approach (\textbf{T-Eval}). The ratio of graph-based
to trace-based time is
shown in the third column (\textbf{Eval-Slowdown}). Since building the graph involves maintaining a heap of
allocated vertices, and we also build the inverse graph at the same time, we expect a larger overhead for the
graph approach. As expected, we do see a higher overhead, being from 2.5x to 15x slower than their
trace-based counterparts. For both approaches, program evaluation always stays within the limit to keep a
user’s attention on the system (10000ms).

{\small\setlength{\tabcolsep}{0.3em}
\begin{longtable}{lrlrlr}
\caption{Average Evaluation Time for Traces vs. Graphs (ms, 1 decimal place)} \label{table:eval-SpdUps} \\
\toprule
   & \multicolumn{2}{c}{\textbf{T-Eval}} & \multicolumn{2}{c}{\textbf{G-Eval}} & \multicolumn{1}{c}{\textbf{Eval-Slowdown}} \\
\midrule
\endfirsthead
\caption[]{Evaluation Time of Traces Versus Graphs (ms, 1 decimal place)} \\
\toprule
   & \multicolumn{2}{c}{T-Eval} & \multicolumn{2}{c}{G-Eval} & \multicolumn{1}{c}{Eval-Slowdown} \\
\midrule
\endhead
\midrule
\multicolumn{6}{r}{Continued on next page} \\
\midrule
\endfoot
\bottomrule
\endlastfoot
\textbf{edge-detect} & \resultLessSquash{732.4}{44.7} & \resultLessSquash{2150.1}{109.9} & 2.94 \\
\textbf{emboss} & \resultLessSquash{633.4}{42.3} & \resultLessSquash{1718.5}{30.9} & 2.71 \\
\textbf{gaussian} & \resultLessSquash{619.9}{48.7} & \resultLessSquash{1714.3}{42.5} & 2.77 \\
\textbf{bar-chart-line-chart} & \resultLessSquash{2336.9}{51.9} & \resultLessSquash{5775.7}{131.6} & 2.47 \\
\textbf{stacked-bar-scatter-plot} & \resultLessSquash{1212.9}{46.2} & \resultLessSquash{7548.9}{476.4} & 6.22 \\
\textbf{grouped-bar-chart} & \resultLessSquash{89.3}{8.4} & \resultLessSquash{1341.6}{132.4} & 15.03 \\
\textbf{line-chart} & \resultLessSquash{130.0}{10.5} & \resultLessSquash{1131.0}{14.4} & 8.70 \\
\textbf{stacked-bar-chart} & \resultLessSquash{55.9}{5.6} & \resultLessSquash{783.0}{15.4} & 14.00 \\
\end{longtable}
}

\subsubsection{Q2: Graph-based demands vs. trace-based demands}

In \tableref{demands-SpdUps}, columns \textbf{T-Demands} and \textbf{G-Demands} show the average times taken
to compute \textit{demands} with the trace and graph approaches, respectively. The ratio of graph-based to
trace-based performance is in the last column, \textbf{Bwd-Speedup}. We compute \textit{demands} over
the graph by running the algorithm for \textit{demBy} from \figref{conjugate:algorithms:demands} over the
opposite graph. In these examples, we keep our query selections small, generally only selecting one or two
output nodes.

Overall, we observe a significant speedup in the performance of \textit{demands} using the graph approach,
from 33x speedup for the \textbf{stacked-bar-chart}, to more than 80x for \textbf{gaussian}. In particular,
every graph-based \textit{demands} query is within the ``instantaneous'' category, whilst all but two of the
examples using the trace approach are in the ``noticable delay'' category. We attribute this speedup to a
couple of factors. First, in the trace approach, each backwards query involves rewinding the entire execution,
regardless of whether all of the trace is relevant; graph queries only traverse the relevant parts of the
graph. Second, in the trace approach, backwards evaluation makes frequent use of lattice join ($\join$) to
combine slicing information from different branches of the computation; joining closure slices in particular
has the potential to be expensive because closures contain environments (which in turn contain closures, and
so on). The graph approach lends itself to a more imperative implementation style, where demand information
accrues against vertices as the backwards analysis proceeds, with no separate join steps.

{\small\setlength{\tabcolsep}{0.4em}
\begin{longtable}{lrlrlc}
\caption{Trace-Based vs.~Graph-Based Implementations of Demands (ms, 1 d.p.)}
\label{table:demands-SpdUps} \\
\toprule
 & \multicolumn{2}{c}{\textbf{T-Demands}} & \multicolumn{2}{c}{\textbf{G-Demands}} & \textbf{Bwd-Speedup} \\
\midrule
\endfirsthead
\caption{Trace-Based vs.~Graph-Based Implementations of Demands (ms, 1 decimal place)} \\
\toprule
 & \multicolumn{2}{c}{\textbf{T-Demands}} & \multicolumn{2}{c}{\textbf{G-Demands}} & \textbf{Demands-Speedup} \\
\midrule
\endhead
\midrule
\multicolumn{6}{r}{Continued on next page} \\
\midrule
\endfoot
\bottomrule
\endlastfoot
\textbf{edge-detect} & \resultLessSquash{61.0}{7.9} & \resultLessSquash{11.9}{0.8} & 50.84 \\
\textbf{emboss} & \resultLessSquash{504.3}{9.1} & \resultLessSquash{8.0}{0.7} & 63.28 \\
\textbf{gaussian} & \resultLessSquash{511.1}{8.3} & \resultLessSquash{6.3}{0.8} & 81.38 \\
\textbf{bar-chart-line-chart} & \resultLessSquash{1274.0}{24.8} & \resultLessSquash{18.5}{0.6} & 68.86 \\
\textbf{stacked-bar-scatter-plot} & \resultLessSquash{616.6}{35.4} & \resultLessSquash{14.0}{1.7} & 44.14 \\
\textbf{grouped-bar-chart} & \resultLessSquash{74.4}{6.2} & \resultLessSquash{1.3}{0.3} & 55.96 \\
\textbf{line-chart} & \resultLessSquash{114.1}{4.8} & \resultLessSquash{1.8}{0.2} & 65.22 \\
\textbf{stacked-bar-chart} & \resultLessSquash{45.1}{5.4} & \resultLessSquash{1.3}{0.4} & 33.62 \\
\end{longtable}
}

\subsubsection{Q3: Various Implementations of Demanded By}

\tableref{equivalent-impls} summarises various implementations of \textit{demanded by}. Column
\textbf{T-DemBy} shows average times for the trace implementation, column \textbf{G-DemBy} shows the graph
implementation using the $\demandedByAlg{}$ algorithm from \figref{conjugate:algorithms:demands}, and column
\textbf{G-DemBy-Suff} shows the graph approach using using the De Morgan dual of the $\sufficesAlg{}$
algorithm from \figref{conjugate:slicing:forward}. Column \textbf{S} shows the speedup of \textbf{G-DemBy}
compared to \textbf{T-DemBy}, and column \textbf{S'} the speedup of \textbf{G-DemBy-Suff} compared to
\textbf{T-DemBy}. As with the \textit{demands} benchmarks, we tend to make small selections, and so the dual
of the $\sufficesAlg{}$ (because it traverses a ``complement'' subgraph) potentially explores a much larger
subgraph than $\demandedByAlg{}$.

For all benchmarks, the graph approach using the $\demandedByAlg{}$ algorithm in
\figref{conjugate:algorithms:demands} performs better than the other two approaches, with five of the
graph-based benchmarks running instantaneously, as opposed to only two for the trace-based approach.

{\small\newcommand{\sideheading}[1]{\scriptsize{\textbf{#1}}}
\setlength{\tabcolsep}{1pt} 
\begin{longtable}{lrlrlcrlc}
\caption{Trace-Based vs.~Graph-Based Implementations of Demanded By (ms, 1 d.p.) where $S= \textit{T-DemBy}/\textit{G-DemBy}$ and $S' = \textit{T-DemBy}/\textit{G-DemBy-Suff}$} \label{table:equivalent-impls} \\
\toprule
& \multicolumn{2}{c}{\textbf{T-DemBy}} & \multicolumn{2}{c}{\textbf{G-DemBy}} & \multicolumn{1}{c}{\textbf{S}} & \multicolumn{2}{c}{\textbf{G-DemBy-Suff}} & \textbf{S'} \\
\midrule
\endfirsthead
\caption{Trace-Based vs.~Graph-Based Implementations of Demanded By (ms, 1 decimal place)} \\
\toprule
   & \multicolumn{2}{c}{T-DemBy} & \multicolumn{2}{c}{G-DemBy} & \multicolumn{1}{c}{S} & \multicolumn{2}{c}{G-DemBy-Suff} & \!\!\!\!$S'$ \\
\midrule
\endhead
\midrule
\multicolumn{6}{r}{Continued on next page} \\
\midrule
\endfoot
\bottomrule
\endlastfoot
\textbf{edge-detect} & \resultLessSquash{696.9}{14.9} & \resultLessSquash{265.8}{19.1} & 2.62 & \resultLessSquash{664.3}{36.3} & 1.05 \\
\textbf{emboss} & \resultLessSquash{603.0}{18.0} & \resultLessSquash{244.9}{16.3} & 2.46 & \resultLessSquash{481.5}{16.4} & 1.25 \\
\textbf{gaussian} & \resultLessSquash{596.8}{11.6} & \resultLessSquash{224.3}{17.4} & 2.66 & \resultLessSquash{474.0}{12.0} & 1.26 \\
\textbf{bar-chart-line-chart} & \resultLessSquash{2079.0}{97.2} & \resultLessSquash{19.1}{1.2} & 108.73 & \resultLessSquash{2360.2}{48.0} & 0.88 \\
\textbf{stacked-bar-scatter-plot} & \resultLessSquash{1004.1}{72.7} & \resultLessSquash{31.3}{1.3} & 32.07 & \resultLessSquash{11598.9}{778.8} & 0.09 \\
\textbf{grouped-bar-chart} & \resultLessSquash{82.4}{9.4} & \resultLessSquash{5.0}{0.6} & 16.47 & \resultLessSquash{2146.5}{158.0} & 0.04 \\
\textbf{line-chart} & \resultLessSquash{113.4}{4.9} & \resultLessSquash{2.3}{0.5} & 50.41 & \resultLessSquash{900.3}{15.3} & 0.13 \\
\textbf{stacked-bar-chart} & \resultLessSquash{42.6}{1.9} & \resultLessSquash{3.2}{1.0} & 13.18 & \resultLessSquash{1271.1}{131.7} & 0.03 \\
\end{longtable}
}

\subsubsection{Discussion of the results}

When comparing the performance of the graph and trace approaches, we observe that the only consistent
performance loss for the graph approach is due to the increased cost of evaluating a program to a graph versus
a trace (\textit{Q1}). For our applications, we argue that this overhead is worth paying; the graph need only
be constructed once, and moreover evaluation time remains within the limit for keeping the user's attention on
the system. More importantly, the fast response times for \textit{demands} and \textit{demanded by} compared
to the trace approach (Q2 and Q3) are good enough for instantaneous queries.

Also notable (observed for Q3), is the performance of \textbf{G-DemBy-Suff}, which uses the dual of the
$\sufficesAlg{}$ algorithm. In the majority of cases, this performs significantly worse than \textbf{G-DemBy}.
The relative difference in performance can perhaps be explained by the fact that our initial selections are
small: since the selection is complemented, and the portion of the graph that the algorithm traverses is in
geneneral much larger. It is also plausible that some of this disparity can be explained by $\sufficesAlg{}$
requiring more book-keeping than $\demandedByAlg{}$.

\section{Related work}
\label{sec:related-work}

\subsection{Data Provenance}

Data provenance research has a long history, ranging from document-level workflow provenance
\cite{callahan06,davidson08} to more fine-grained database techniques like \emph{where} provenance
\cite{buneman01}, with comparatively less emphasis on general-purpose languages. \citet{dietrich22} track both
value and control dependencies in recursive database queries and user-defined functions, implementing their
approach by rewriting SQL queries to provenance-tracking counterparts; \citet{fehrenbach19} explore provenance
for language-integrated query, extending the multi-tier language Links \cite{cooper06} with provenance
tracking, but only for the embedded SQL part of the language. In terms of fine-grained techniques for
structured outputs like visualisations, the most closely related lines of work are \citet{psallidas18smoke}
for relational languages, and \citet{perera22} for general-purpose languages. These authors all emphasise the
importance of \emph{cognacy}, i.e.~common ancestry, albeit not in the context of dependence graphs, showing
how to ``link'' outputs by tracing back from an output selection to a data selection and then forward to
another output selection. These relations of common ancestry seem to be largely unstudied in the data
provenance literature, despite their importance in data visualisation.

\subsection{Dynamic Dependence Graphs}

Dynamic dependence graphs have been used extensively for program slicing \cite{field98,hammer06}, optimisation
\cite{ferrante87}, incremental computation \cite{acar02} and fault localisation~\cite{Soremekun2021}. Most of
these applications are for imperative languages, where it is useful distinguish between control and data
edges; because \OurLang is pure, we use somewhat simpler dependence graphs with a single kind of edge. However
different kinds of edge are likely to be needed for future applications (\secref{future}). The main
contribution of our work to the dependence graph paradigm are cognacy operators which allow the identification
of minimally related sets of inputs ($\relInputR$) and minimally related sets of outputs ($\relOutputR$), in a
formal setting where we show these operators to be self-conjugate and related to Galois connections.

\subsection{Galois Slicing}

Although we do not consider program slices in the present work, our formal setting is closely related to
program slicing techniques based on Galois connections \cite{perera12a,perera16d,ricciotti17}. In these works,
which use a bidirectional interpreter, the forward and backward directions of the interpreter correspond to
our \emph{sufficiency} ($\sufficesR$) and \emph{demands} ($\demandsR$) queries. \citet{perera22} extended this
approach to a setting where slices have complements, and showed how composing the backward analysis with the
De Morgan dual of the forwards analysis can be used to implement linked brushing. As well as requiring a
bidirectional interpreter, the main difficulty with this approach has been achieving interactive performance.
Implementations have relied on a sequential execution trace, used in the forwards direction to ``replay'' the
computation, propagating sufficiency information from inputs to outputs, and in the backwards direction to
``rewind'' the computation, propagating demand from outputs to inputs. As discussed in \secref{evaluation},
and in contrast to the graph approach presented here, this approach is overly sequential and unable to exploit
the independence of some subcomputations that makes slicing useful in the first place.

\subsection{Linked Visualisations}

The connection between visual selection (such as clicking or lasooing) and data ``selection'' has been
explored in the data visualisation literature, leading to various techniques for inverting selections in
visual space to obtain selections in the underlying data~\cite{north00, heer08}. There is also a substantial
literature on linked brushing~\cite{becker87, livny97} as a way of connecting output selections via mediating
data selections. Although this has long been recognised as an important visual tool, with \citet{roberts06}
arguing it should be ubiquitous, there is relatively little work on general-purpose infrastructure to support
it. For example Reactive Vega~\cite{satyanarayan17} provides a rich set of interactivity features to support
linked brushing, but no mechanism for automatically propagating selections through arbitrary intervening
queries; Glue \cite{beaumont13} is a powerful Python library for building linked visualisations for scientific
applications, but the developer is responsible for specifying the relationships between data sets that unpin
the linking. More infrastructural approaches to linked brushing \cite{north00,livny97,psallidas18smoke},
similar in spirit to our work, have been developed mainly for relational languages. Relational languages are a
promising direction in data visualisation, but general-purpose languages continue to be widely used too, so it
is important to support linked brushing in this context as well.

\subsection{Transparent research outputs}

\citet{dragicevic19} also develop techniques for transparent research outputs, via explorable multiverse
analysis, exposing \emph{analysis parameters}, the methodological choices made in designing a data analysis.
Since different methodological choices will impact the conclusions of the analysis, they expose these choices
to users, and allow them to switch between different choices and observe the impact on the results. Our work
is potentially complementary: we are concerned with surfacing the dependencies that arise within a specific
choice of analysis parameters, whereas multiverse analyses are about exploring how things change under
different choices. It might be useful to package our dependency analysis as part of such a multiverse
analysis, allowing the user to observe the different dependency structures that arise from different
methodological decisions.

\section{Conclusion and Future Work}
\label{sec:conclusion}

Visualisations are an essential tool for communicating science and other data-driven claims, but can be hard
to make sense of and trust. Visualisations and other outputs that are ``transparent'' -- that can reveal to an
interested reader how they are related to data -- are more informative and trustworthy, but are also difficult
to produce. In this paper, we introduced a novel bidirectional program analysis framework for a
general-purpose programming language that makes it easier to create transparent visualisations by shifting
much of the burden to the language runtime. Relative to prior work based on execution traces and bidirectional
interpreters, our approach also makes life easier for the language implementor, by decomposing the system into
a single evaluator that builds a dependence graph, and a pair of language-agnostic bidirectional graph
algorithms. We also showed an overall performance improvement for provenance queries, at the cost of some
overhead in building the dependence graph.

\subsection{Future work}
\label{sec:future}

We identify two important directions for future work. First, the dependency relation captured by the
dependence graph semantics in \secref{core} omits certain intuitively plausible edges, such as the
\emph{projection} edge that one might expect to exist between a record and the value of a contained field as a
consequence of evaluating a field access expression $\exRecProj{e}{x}$. On the other hand, recording
\emph{all} structural dependencies of this nature would significantly bloat the dependence graph with
information that is only relevant in certain contexts (for example when one is specifically concerned with
where a value came from, rather than how it was computed). We would like to develop a more semantically
justified dependence graph, along with a proof that the graph is ``complete'' in some formal sense (cf.~the
\emph{dependency correctness} notion of \citet{cheney11b}), but anticipate that making this richer graph
practical will require distinguishing different kinds of query (e.g.~``how'' vs.~``where from'') for use in
different contexts.

Second, although the graph algorithms presented in \secref{conjugate} compute graph slices as an interim data
structure, the internal nodes of the graph are eventually discarded. This is fine for the use cases in
\secref{overview}, which only concern ``extensional'' (IO) transparency. However, intensional information
would be highly informative too: if highlighting a point in the moving average in
\figref{introduction:related-inputs} could show not only that three emissions inputs were relevant, but that
those values were summed and then divided by three, then the chart alone would serve as a full
``self-explanatory'' implementation of moving average.

This would connect to work on \emph{executable program slicing} \cite{field98}, and there are substantial
challenges that arise here as well: in particular, intensional explanations get potentially very large, so
ways of reducing and managing their complexity will be important. Again \citet{cheney11b} offer inspiration:
their idea of \textit{expression provenance} may be a way of presenting pared-back but still informative
explanations (for example showing the tree of primitive arithmetic operations that computed a value, but
omitting user-defined function calls and branches). And the user may only want intensional information for
certain subcomputations, and would be happy with just IO relationships for others --- in which case graph
transformations which elide internal nodes but preserve IO connectivity, such as the Y-$\Delta$ transform
\cite{truemper89}, may also have a role to play.

\printbibliography

\clearpage
\appendix
\section{Proofs of theorems from \secref{conjugate}}

\subsection{\lemref{im-preim-conjugate}}
For any relation $R \subseteq X \times Y$ the functions $\demandedByR_{R}$ and $\demandsR_{R}$ are conjugate.
\setcounter{equation}{0}
\proofContext{im-preim-conj}
\begin{proof}
   Let $X' \in \powerset{X}$, $Y' \in \powerset{Y}$.
   \small
   \newline
   \textbf{WTS:} $\demandedByR_{R}(X') \cap Y' = \emptyset \implies X' \cap \demandsR_{R}(Y') = \emptyset$
   \begin{flalign}
      &
      \demandedByR_{R}(X') \cap Y' = \emptyset
      &
      \text{suppose}
      \notag
      \\
      &
      \set{y \in Y \mid \exists x \in X'.(x,y) \in R} \cap Y' = \emptyset
      &
      \text{def of $\demandedByR_{R}$}
      \notag
      \\
      &
      \set{y \in Y' \mid \exists x \in X'.(x,y) \in R} = \emptyset
      &
      \text{def of $\cap$}
      \locallabel{preim-intersection}
      \\
      &
      \forall y \in Y' .\; \forall x \in X'. (x,y) \not\in R
      &
      \text{equivalent to (\localref{preim-intersection})}
      \locallabel{unintersect}
      \\
      &
      \set{x \in X' \mid \exists y \in Y'.(x,y) \in R} = \emptyset
      &
      \text{equivalent to (\localref{unintersect})}
      \notag
      \\
      &
      \qedLocal
      X' \cap \set{x \in X \mid \exists y \in Y'.(x,y) \in R} = \emptyset
      &
      \text{def of $\cap$}
      \notag
      \\
      &
      X' \cap \demandsR_{R}(Y') = \emptyset 
      &
      \text{def of $\demandsR_{R}$}
      \notag
   \end{flalign}
   \newline
   \textbf{WTS: } $ X' \cap \demandsR_{R}(Y') = \emptyset \implies \demandedByR_{R}(X') \cap Y' = \emptyset$
   \begin{flalign}
      &
      X' \cap \demandsR_{R}(Y') = \emptyset
      &
      \text{suppose}
      \notag
      \\
      &
      X' \cap \set{x \in X \mid \exists y \in Y'.(x,y) \in R} = \emptyset
      &
      \text{def of $\demandsR_{R}$}
      \notag
      \\
      &
      \set{x \in X' \mid \exists y \in Y'.(x,y) \in R}
      &
      \text{def of $\cap$}
      \locallabel{preim-intersect-2}
      \\
      &
      \forall y \in Y' .\; \forall x \in X'. (x,y) \not\in R
      &
      \text{equivalent to (\localref{preim-intersect-2})}
      \locallabel{preim-unintersect-2}
      \\
      &
      \qedLocal
      \set{y \in Y \mid \exists x \in X'.(x,y) \in R} \cap Y' = \emptyset
      &
      \text{equivalent to (\localref{preim-unintersect-2})}
      \notag
      \\
      &
      \demandedByR_{R}(X') \cap Y' = \emptyset
      &
      \text{def of $\demandedByR_{R}$}
      \notag
   \end{flalign}
\end{proof}

\subsection{\propref{conj-galois-correspond}}
Suppose functions $f: A \to B$ and $g: B \to A$ between Boolean algebras. The following statements are
equivalent:
\begin{enumerate}
   \item $f$ and $g$ form a conjugate pair
   \item $f$ and $\dual{g}$ form a Galois connection
\end{enumerate}
\setcounter{equation}{0}
\proofContext{galois-conj}
\begin{proof}
   Fix $x \in A$ and $y \in B$.
   \newline
   \textbf{WTS: }$f, g$ conjugate $\implies$ $f, \dual{g}$ form Galois connection.
   \begin{flalign}
      &
      f(x) \leq_{B} y
      &
      \text{(suppose)}
      \notag
      \\
      &
      f(x) \wedge_{B} \neg y = \bot
      &
      \notag
      \\
      &
      x \wedge_{A} g(\neg y) = \bot
      &
      \text{(conjugacy)}
      \notag
      \\
      &
      x \leq_{A} \neg (g(\neg y))
      &
      \notag
      \\
      &
      \qedLocal
      x \leq_{A} \dual{g}(y)
      &
      \text{(def of $\dual{\cdot}$)}
      \notag
   \end{flalign}
   \newline
   \textbf{WTS: }$f, \dual{g}$ Galois connection $\implies$ $f, g$ conjugate
   \begin{flalign}
      &
      f(x) \wedge_{B} y = \bot
      &
      \text{suppose}
      \notag
      \\
      &
      f(x) \leq_{B} \neg y 
      &
      \notag
      \\
      &
      x \leq_{A} \dual{g}(\neg y)
      &
      \text{(GC)}
      \notag
      \\
      &
      x \leq_{A} \neg (g (y))
      &
      \text{(def of $\dual{\cdot}$)}
      \notag
      \\
      &
      x \wedge_{A} g(y) = \bot
      &
      \notag
   \end{flalign}
\end{proof}
\subsection{\lemref{conjugate:image-preimage-duality}}
\label{sec:apdx:conjugate:image-preimage-duality}
\setcounter{equation}{0}
\proofContext{conjugate-galois-duality}
Fix $R \subseteq X \times Y$. We proceed by direct calculation.
\begin{proof}
   \textbf{WTS: }$\dual{\demandedByR_R} = \sufficesR_R$:
   \begin{flalign}
      &
      \dual{\demandedByR_R (X')} = \dual{\set{y \in Y \mid \exists x \in X'.(x,y) \in R}}
      &
      \text{def of $\demandedByR_R$}
      \notag
      \\
      &
      = \compl{\set{y \in Y \mid \exists x \in \compl{(X')}.(x,y) \in R}}
      &
      \text{def of $\dual{\cdot}$}
      \notag
      \\
      &
      = \compl{\set{y \in Y \mid \exists x \in (X \setminus X').(x, y) \in R}}
      &
      \text{def of $\compl{\cdot}$}
      \notag
      \\
      &
      = \set{y \in Y \mid \nexists x \in (X \setminus X').(x, y) \in R}
      &
      \text{$\compl{\cdot}$ negates predicate}
      \notag
      \\
      &
      = \sufficesR_R(X')
      &
      \text{def of $\sufficesR_R$}
      \notag
   \end{flalign}
   \newline
   \textbf{WTS: }$\dual{\demandsR_R} = \preimageDualR_R$:
   \begin{flalign}
      &
      \dual{\demandsR_R (Y')} = \dual{\set{x \in X \mid \exists y \in Y'. (x, y) \in R}}
      &
      \text{def of $\dual{\demandsR_R}$}
      \notag
      \\
      &
       = \compl{\set{x \in X \mid \exists y \in \compl{(Y')}. (x, y) \in R}}
      &
      \text{def of $\dual{\cdot}$}
      \notag
      \\
      &
      = \compl{\set{x \in X \mid \exists y \in (Y \setminus Y'). (x, y) \in R}}
      &
      \text{def of $\compl{\cdot}$}
      \notag
      \\
      &
       = \set{x \in X \mid \nexists y \in (Y \setminus Y').(x, y) \in R}
      &
      \text{$\compl{\cdot}$ negates predicate}
      \notag
      \\
      &
      = \preimageDualR_R(Y')
      &
      \text{def of $\preimageDualR_R$}
      \notag
   \end{flalign}
\end{proof}

\subsection{\propref{demandedByAlg-computes-demandedBy}}
\label{sec:apdx:demandedBy}
\setcounter{equation}{0}

\proofContext{demandedByAlg-computes-demandedBy}
\begin{proof}
   \small
   \begin{flalign}
      &
      \caseName{$\Rightarrow$ direction:}
      \notag
      \\
      &
      X \subseteq \sources{G}
      \text{ with }
      X \demandedByAlg{G} Y
      &
      \text{suppose}
      \locallabel{X-sources}
      \\
      &
      \derivation{\derivationWidth}{
         \begin{smathpar}
            \inferrule*
            {
               (X, \emptyset), G \demandedByV{} H
            }
            {
               X \demandedByAlg{G} \sinks{H}\Lowlight{\;= Y}
            }
         \end{smathpar}
      }
      &
      \text{inversion; $\exists H$}
      \notag
      \\
      &
      \qedLocal
      \demandedByR_G(X) = \sinks{G} \cap \demandedByR^*_{G}(X) = Y
      &
      \text{(\localref{X-sources}); \ref{lem:appendix:proof:conjugate:algorithms:demanded-cap-sinks}; \propref{appendix:proof:conjugate:algorithms:correctness-demByV}}
      \notag
      \\
      &
      \caseName{$\Leftarrow$ direction:}
      \notag
      \\
      &
      X \subseteq \sources{G}\text{ with }\demandedByR_G(X) = \sinks{G} \cap \demandedByR^*_G(X) = Y
      &
      \text{suppose}
      \notag
      \\
      &
      (X, \emptyset), G \demandedByV{} H
      \text{ with }\sinks{H} = Y
      &
      \text{\propref{appendix:proof:conjugate:algorithms:correctness-demByV}; $\exists H$}
      \notag
      \\
      &
      \qedLocal
      \derivation{\derivationWidth}{
         \begin{smathpar}
            \inferrule*
            {
               (X, \emptyset), G \demandedByV{} H
            }
            {
               X \demandedByAlg{G} \sinks{H}\Lowlight{\;= Y}
            }
         \end{smathpar}
      }
      &
      \text{def.~$\demandedByAlg{}$}
      \notag
   \end{flalign}
\end{proof}

\subsection{Correctness of $\demandedByV{}$}

The \emph{output} relation of a graph $G$ with reachability relation $R$, defined as $R \cap \V(G) \times
\sinks{G}$, gives rise to the following extension of $\demandedByR_{G}$:

\begin{definition}[$\demandedByR^*$ for a dependence graph]
   For a graph $G$ with output relation $R$, define:
   \[\demandedByR^*_G \eqdef \demandedByR_{R}: \powerset{\V(G)} \to \powerset{\V(G)}\]
\end{definition}

\begin{lemma}
   \label{lem:appendix:proof:conjugate:algorithms:demanded-cap-sinks}
   $\demandedByR_{G} = \sinks{G} \cap \demandedByR^*_{G}(X)$.
\end{lemma}

\begin{lemma}
   \label{lem:appendix:proof:conjugate:algorithms:sinks-fixpoint}
   If $X \subseteq \sinks{G}$ then $\demandedByR^*_{G}(X) = X$.
\end{lemma}

\begin{lemma}
   \label{lem:appendix:proof:conjugate:algorithms:edge-conservation}
   Suppose $\V(H) \subseteq \V(G)$. For any $\alpha \in \sinks{H}$ and any
   $E = \outE{G}(\alpha) \neq \emptyset$:
   $$
      \demandedByR^*_G(\sinks{H}) = \demandedByR^*_{G \edgeMinus E} (\sinks{H \cup E})
   $$
\end{lemma}

\begin{proposition}
   \label{prop:appendix:proof:conjugate:algorithms:correctness-demByV}
   For $\V(H) \subseteq \V(G)$, and any $Y \subseteq \sinks{G}$:
   \[(\exists H'.\;H, G \demandedByV{} H'\text{ with }\sinks{H'} = Y) \iff \sinks{H \cup G} \cap \demandedByR^*_G(\sinks{H}) = Y\]
\end{proposition}

\proofContext{correctness-demByV}
\begin{proof}
   \small
   \begin{flalign}
      &
      \V(H)\subseteq \V(G)\text{ and }Y \subseteq \sinks{G}
      &
      \text{suppose}
      \notag
      \\
      &
      \caseName{$\Rightarrow$ direction:}
      \notag
      \\
      &
      H, G \demandedByV{} H'\text{ with }\sinks{H'} = Y
      &
      \text{suppose; $\exists H'$}
      \notag
      \\
      &
      \text{By induction on derivation of $\demandedByV{}$.}
      &
      \notag
      \\
      &
      \subcaseDerivation{\derivationWidth}{
         \begin{smathpar}
            \inferrule*[
            lab={\ruleName{done}},
            right={$\sinks{H} \subseteq \sinks{G}$}
            ]
            {
               \strut
            }
            {
               H, G \demandedByV{} H\Lowlight{\;=H'}
            }
         \end{smathpar}
      }
      &
      \notag
      \\
      &
      \qedLocal
      \sinks{H \cup G} \cap \demandedByR^*_G(\sinks{H}) = \sinks{H} \cap \demandedByR^*_{G}(\sinks{H}) = \sinks{H}
      &
      \text{\lemref{appendix:proof:conjugate:algorithms:sinks-fixpoint}}
      \notag
      \\
      &
      \subcaseDerivation{\derivationWidth}{
         \begin{smathpar}
            \inferrule*[
               lab={\ruleName{extend}},
               right={$\alpha \in \sinks{H} \wedge E = \outE{G}(\alpha) \neq \emptyset$}
            ]
            {
               H \cup E, G \edgeMinus E
               \demandedByV{}
               H'
            }
            {
               H, G \demandedByV{} H'
            }
         \end{smathpar}
      }
      &
      \notag
      \\
      &
      \demandedByR^*_{G \edgeMinus E}(\sinks{H \cup E}) = \sinks{H'}
      &
      \text{IH}
      \notag
      \\
      &
      \qedLocal
      \demandedByR^*_{G}(\sinks{H}) = \sinks{H'}
      &
      \text{\lemref{appendix:proof:conjugate:algorithms:edge-conservation}}
      \notag
      \\
      &
      \caseName{$\Leftarrow$ direction:}
      \notag
      \\
      &
      \text{By strong induction over $\length{\E(G)}$.}
      &
      \notag
      \\
      &
      \sinks{G \cup H} \cap \demandedByR^*_G(\sinks{H}) = Y
      &
      \text{suppose}
      \locallabel{assumption}
      \\
      &
      \subcase{\sinks{H} \subseteq \sinks{G}}
      \notag
      \\
      &
      \derivation{\derivationWidth}{
         \begin{smathpar}
            \inferrule*[
            lab={\ruleName{done}},
            right={$\sinks{H} \subseteq \sinks{G}$}
            ]
            {
               \strut
            }
            {
               H, G \demandedByV{} H
            }
         \end{smathpar}
      }
      &
      \text{def.~$\demandedByV{}$}
      \notag
      \\
      &
      \qedLocal
      \sinks{G \cup H} \cap \sinks{H} = \sinks{G \cup H} \cap \demandedByR^*_{G}(\sinks{H}) = Y
      &
      \text{\lemref{appendix:proof:conjugate:algorithms:sinks-fixpoint}; (\localref{assumption})}
      \notag
      \\
      &
      \subcase{\alpha \in \sinks{H} \wedge E = \outE{G}(\alpha) \neq \emptyset}
      &
      \text{$\exists \alpha$, $\exists E$}
      \notag
      \\
      &
      \E(G \edgeMinus E) \subset \E(G)
      &
      \notag
      \\
      &
      H \cup E, G \edgeMinus E \demandedByV{} H'
      \text{ with }
      \sinks{H'} = \sinks{(G \edgeMinus E) \cup (H \cup E)} \cap \demandedByR^*_{G \edgeMinus E}(\sinks{H \cup E})
      &
      \text{IH; $\exists H'$}
      \notag
      \\
      &
      \sinks{H'} = \sinks{G\cup H} \cap \demandedByR^*_{G}(\sinks{H}) = Y
      &
      \text{\lemref{appendix:proof:conjugate:algorithms:edge-conservation}}
      \notag
      \\
      &
      \qedLocal
      \derivation{\derivationWidth}{
         \begin{smathpar}
            \inferrule*[
               lab={\ruleName{extend}},
               right={$\alpha \in \sinks{H} \wedge E = \outE{G}(\alpha) \neq \emptyset$}
            ]
            {
               H \cup E, G \edgeMinus E
               \demandedByV{}
               H'
            }
            {
               H, G \demandedByV{} H'
            }
         \end{smathpar}
      }
      &
      \text{def.~$\demandedByV{}$}
      \notag
   \end{flalign}
\end{proof}

\subsection{\propref{sufficesAlg-computes-suffices}}
\setcounter{equation}{0}
\label{sec:apdx:suffices-alg}
\proofContext{sufficesAlg-suffices}
\begin{proof}
   \small
   \begin{flalign}
      &
      \caseName{$\Rightarrow$ direction:}
      &
      \notag
      \\
      &
      X \subseteq \sources{G}
      \text{ and }
      Y \subseteq \sinks{G}
      \text{ with }
      X \sufficesAlg{G} Y
      &
      \text{suppose}
      \notag
      \\
      &
      \derivation{\derivationWidth}{
         \begin{smathpar}
            \inferrule*
            {
               (X,\emptyset), (\emptyset,\emptyset), G \sufficesE{} H
            }
            {
               X \sufficesAlg{G} \V(H) \cap \sinks{G} \Lowlight{\;=Y}
            }
         \end{smathpar}
      }
      &
      \text{inversion; $\exists H$}
      \notag
      \\
      &
      (X, \emptyset), (\emptyset,\emptyset), G
      \text{ form a partial slice of } G
      &
      \text{immediate}
      \notag
      \\
      &
      \qedLocal
      \sufficesR_{G}(X) = \sinks{G} \cap \sufficesR^*_{G}(X)= Y
      &
      \text{\lemref{appendix:proof:conjugate:algorithms-suff:star-intersect-sinks}, \propref{appendix:proof:conjugate:algorithms-suff:correctness-sufficesE}}
      \notag
      \\
      &
      \caseName{$\Leftarrow$ direction: }
      \notag
      \\
      &
      X \subseteq \sources{G}
      \text{ and }
      Y \subseteq \sinks{G}
      \text{ with }
      \sufficesR_{G}(X) = Y
      &
      \text{suppose}
      \notag
      \\
      &
      \sufficesR_{G}(X) = \sinks{G} \cap \sufficesR^*_{G}(X)
      &
      \text{\lemref{appendix:proof:conjugate:algorithms-suff:star-intersect-sinks}}
      \notag
      \\
      &
      (X, \emptyset), (\emptyset,\emptyset), G
      \text{ form a partial slice of } G
      &
      \text{immediate}
      \notag
      \\
      &
      (X, \emptyset), (\emptyset,\emptyset), G \sufficesE{} H
      &
      \text{\propref{appendix:proof:conjugate:algorithms-suff:correctness-sufficesE}; $\exists H$}
      \notag
      \\
      &
      \qedLocal
      \derivation{\derivationWidth}{
         \begin{smathpar}
            \inferrule*
            {
               (X,\emptyset), (\emptyset,\emptyset), G \sufficesE{} H
            }
            {
               X \sufficesAlg{G} \V(H) \cap \sinks{G} \Lowlight{\;=Y}
            }
         \end{smathpar}
      }
      &
      \text{def.~$\sufficesAlg{}$}
      \notag
   \end{flalign}
\end{proof}
\subsection{Correctness of $\sufficesE{}$}

\begin{definition}[Partial Slice]
   \label{def:partial-slice}
   Graphs $H,P,G$ constitute a \textit{partial slice} of the graph $G_0$ iff $H,P,G$ form a partition of $G_0$ and:
   \begin{enumerate}[label=(\roman*)]
      \item $\V(H) \subseteq \V(P) \subseteq \V(G) = \V(G_0)$
      \item $\V(H) \supseteq \compl{\sinks{P}}$
      \item $\V(H) = \sources{P}$
      \item $\V(P) = \sources{P} \cap \sinks{P}$
      \item $\forall (\alpha,\beta) \in (\sources{G_0} \setminus \sources{H}) \times \V(H).\;(\alpha,\beta) \notin R_{G_0}$
      \item $\sources{G}\setminus(\sinks{P}\setminus\sources{P}) = \V(H) \uplus (\sources{G_0}\setminus\sources{H})$
   \end{enumerate}
\end{definition}
\begin{corollary}
   $H,P,G$ a partial slice $\iff$ $H,P,G$ edge disjoint and $\forall (\alpha,\beta) \in \E(G) \cup \E(P).\;\beta \notin \V(H)$
\end{corollary}
\begin{definition}[Input Relation]
   For a graph $G$ with reachability relation $R$, define it's \text{input relation}
   $R_{I} \eqdef (\sources{G} \times \V(G)) \cap R$
\end{definition}

\begin{definition}[$\sufficesR^*$ for a dependence graph]
   For a graph $G$, with an input relation $R_{I}$, and a set of input vertices $X \subseteq \sources{G}$:
   \[\sufficesR^*_{G}(X) \eqdef \sufficesR_{R_{I}}(X) = \set{y \in \V(G) \mid \nexists x \in \sources{G} \setminus X.\; (x,y) \in R_{I}}\]
\end{definition}

\begin{lemma}
   \label{lem:appendix:proof:conjugate:algorithms-suff:star-intersect-sinks}
   $\sufficesR_{G}(X) = \sinks{G} \cap \sufficesR^*_{G}(X)$
\end{lemma}
\begin{lemma}
   \label{lem:appendix:proof:conjugate:algorithms-suff:suffices-sources}
   If $H,P,G$ a partial slice of $G_0$ with $\sources{G} \cap \V(P) \subseteq \sources{P}$ and $\V(H) \subseteq \sinks{G}$:
   \[ \sufficesR^*_{G_0}(\sources{H}) = \V(H) \]
\end{lemma}
\begin{proof}
   \small
   \begin{flalign}
      &
      H,P,G \text{ a partial slice}
      &
      \text{suppose}
      \notag
      \\
      &
      \sources{G} \cap \V(P) \subseteq \sources{P}
      &
      \text{suppose}
      \locallabel{sources-p-subset-sources-g}
      \\
      &
      \V(H) \subseteq \sinks{G}
      &
      \text{suppose}
      \locallabel{sinks-h-subset-sinks-g}
      \\
      &
      \caseName{$\subseteq$ direction}
      \notag
      \\
      &
      y \in \sufficesR^*_{G_0}(\sources{H})
      &
      \text{suppose}
      \notag
      \\
      &
      \forall x \in \sources{G_0} \setminus \sources{H}.\; (x,y) \notin R_{G_0}
      &
      \text{def. $\sufficesR^*$}
      \notag
      \\
      &
      y \in \V(G)
      &
      \notag
      \\
      &
      x \in \sources{G}.\;(x,y) \in R_{G}
      &
      \text{$\exists x$, $G$ acyclic}
      \notag
      \\
      &
      x \in \V(H) \cup \sources{G_0} \setminus \sources{H} \cup \sinks{P} \setminus \sources{P}
      &
      \text{$H,P,G$ a partial slice}
      \notag
      \\
      &
      x \notin \sources{G_0} \setminus \sources{H}
      &
      \text{immediate}
      \notag
      \\
      &
      x \notin \sinks{P} \setminus \sources{P}
      &
      \text{contradicts \localref{sources-p-subset-sources-g}}
      \notag
      \\
      &
      \qedLocal
      x \in \V(H)
      &
      \text{$\sources{G}$ invariant}
      \notag
      \\
      &
      x \neq y \implies x \notin \sinks{H} \implies \bot
      &
      \text{contradicts \localref{sinks-h-subset-sinks-g}}
      \notag
      \\
      &
      \qedLocal
      x = y
      &
      \notag
      \\
      &
      \qedLocal y \in \V(H)
      &
      \notag
      \\
      &
      \caseName{$\supseteq$ direction}
      \notag
      \\
      &
      y \in \V(H)
      &
      \text{suppose}
      \notag
      \\
      &
      y \in \V(G_0)
      &
      \notag
      \\
      &
      \sources{G_0} \setminus \sources{H} \subseteq (\V(G \cup P) \setminus \V(H))
      &
      \notag
      \\
      &
      \forall \alpha \in \sources{G_0} \setminus \sources{H}.\; (\alpha,y) \notin R_{G_0}
      &
      \text{def. partial slice}
      \notag
      \\
      &
      \qedLocal
      y \in \sufficesR^*_{G_0}(\sources{H})
      &
      \notag
   \end{flalign}
\end{proof}
\begin{lemma}
   \label{lem:appendix:proof:conjugate:algorithms-suff:partial-slice-preserved}
   \item
   \begin{enumerate}
      \item If $\alpha \in \V(H)$, $E = \outE{G}(\alpha) \neq \emptyset$ then: \newline $H,P,G$ a partial slice of $G_0$ $ \iff H,P \cup E,G \edgeMinus E$ a partial slice of $G_0$
      \item If $\alpha \in \sources{G}$, $E = \inE{P}(\alpha) \neq \emptyset$ then: $H,P,G$ a partial slice $G_0$ $\iff H \cup E, P \edgeMinus E, G$ a partial slice of $G_0$
   \end{enumerate}
\end{lemma}
\proofContext{partial-slice-lem}
\begin{proof}
   \small
   \begin{flalign}
      &
      \caseName{(1.II), $\Rightarrow$}
      \notag
      \\
      &
      E = \set{\alpha} \times B
      &
      \text{suppose}
      \notag
      \\
      &
      \compl{\sinks{P}} \subseteq \V(H)
      &
      \text{suppose}
      \notag
      \\
      &
      \compl{\sinks{P \cup E}} = \compl{\sinks{P}} \cup \set{\alpha}
      &
      \text{def. $E$}
      \notag
      \\
      &
      \qedLocal
      \compl{\sinks{P \cup E}} \subseteq \V(H)
      &
      \text{$\alpha \in \V(H)$}
      \notag
      \\
      &
      \caseName{(1.II), $\Leftarrow$}
      \notag
      \\
      &
      \compl{\sinks{P \cup E}} \subseteq \V(H)
      &
      \text{suppose}
      \notag
      \\
      &
      \compl{\sinks{P}} \subseteq \compl{\sinks{P \cup E}}
      &
      \notag
      \\
      &
      \qedLocal
      \compl{\sinks{P}} \subseteq \V(H)
      &
      \notag
      \\
      &
      \caseName{(1.III),$\Rightarrow$}
      \notag
      \\
      &
      E = \set{\alpha} \times {B}
      &
      \text{suppose}
      \notag
      \\
      &
      \V(H) = \sources{P}
      &
      \text{suppose}
      \notag
      \\
      &
      \text{Need to show: }
      \V(H) = \sources{P \cup E}
      &
      \notag
      \\
      &
      \sources{P \cup E} = \sources{P}\setminus(\sources{P}\cap B)
      &
      \notag
      \\
      &
      \text{Need to show: }
      \sources{P}\cap B = \emptyset
      &
      \notag
      \\
      &
      (\alpha, \beta) \in \E(G)
      &
      \text{$\forall (\alpha,\beta) \in E$}
      \notag
      \\
      &
      \beta \in \sources{P} = \V(H) \implies \bot
      &
      \text{Contra partial slice(V)}
      \notag
      \\
      &
      \qedLocal
      \beta \notin \sources{G}
      &
      \notag
      \\
      &
      \qedLocal
      \sources{P} \cap \beta = \emptyset
      &
      \notag
      \\
      &
      \qedLocal
      \V(H) = \sources{P \cup E}
      &
      \notag
      \\
      &
      \caseName{(1.III), $\Leftarrow$}
      \notag
      \\
      &
      \V(H) = \sources{P \cup E}
      &
      \text{suppose}
      \notag
      \\
      &
      E = \set{\alpha} \times B
      &
      \text{suppose}
      \notag
      \\
      &
      \text{Want to show: }
      \V(H) = \sources{P}
      &
      \notag
      \\
      &
      \sources{P} = \sources{P \cup E} \setminus B
      &
      \notag
      \\
      &
      \text{suffices to show: }
      \sources{P \cup E} \cap B = \emptyset
      &
      \notag
      \\
      &
      (\alpha, \beta) \in E.\; (\alpha,\beta) \in \E(G)
      &
      \notag
      \\
      &
      \qedLocal
      \beta \notin \V(H) = \sources{P \cup E}
      &
      \text{partial slice (V)}
      \notag
      \\
      &
      \qedLocal
      \sources{P \cup E} \cap B = \emptyset
      &
      \notag
      \\
      &
      \qedLocal
      \sources{P} = \sources{P \cup E}
      &
      \notag
      \\
      &
      \qedLocal
      \V(H) = \sources{P}
      &
      \notag
      \\
      &
      \caseName{(1.IV), $\Rightarrow$}
      \notag
      \\
      &
      \V(P) = \sources{P} \cup \sinks{P}
      &
      \text{suppose}
      \notag
      \\
      &
      E = B \times \set{\alpha}
      &
      \text{suppose}
      \notag
      \\
      &
      \V(P \cup E) = \V(P) \cup B
      &
      \notag
      \\
      &
      \text{suffices to show that:}
      \sources{P}\cup\sinks{P}\cup B = \sources{P \cup E} \cup \sinks{P \cup E}
      &
      \notag
      \\
      &
      \sources{P} = \sources{P \cup E}
      &
      \text{proved in proof of (iii)}
      \notag
      \\
      &
      \alpha \in \sources{P} = \sources{P \cup E}
      &
      \text{$\alpha \in \V(H)$}
      \notag
      \\
      &
      \sinks{P \cup E} = (\sinks{P} \cup B)\setminus {\alpha}
      &
      \text{def. $E$}
      \notag
      \\
      &
      \qedLocal
      \sources{P} \cup \sinks{P} \cup B = \sources{P \cup E} \cup ((\sinks{P} \cup B)\setminus \set{\alpha})
      &
      \notag
      \\
      &
      \qedLocal
      \sources{P} \cup \sinks{P} \cup B = \sources{P \cup E} \cup \sinks{P \cup E}
      &
      \notag
      \\
      &
      \caseName{(1.IV), $\Leftarrow$}
      \notag
      \\
      &
      \V(P \cup E) = \sources{P \cup E} \cup \sinks{P \cup E}
      &
      \text{suppose}
      \notag
      \\
      &
      \V(P \cup E) = \V(P) \cup B
      &
      \notag
      \\
      &
      \sources{P \cup E} \cup \sinks{P \cup E} = \V(P) \cup B
      &
      \notag
      \\
      &
      \sources{P \cup E} = \sources{P}
      &
      \text{def. $E$}
      \notag
      \\
      &
      \sinks{P \cup E} = (\sinks{P} \cup B) \setminus \set{\alpha}
      &
      \text{def. $E$}
      \notag
      \\
      &
      B \cap \sources{P} = \emptyset
      &
      \notag
      \\
      &
      \sources{P} \setminus B = \sources{P}
      &
      \notag
      \\
      &
      (\sources{P \cup E} \cup \sinks{P \cup E})\setminus B = \V(P)
      &
      \notag
      \\
      &
      (\sources{P \cup E}\setminus B) \cup ((\sinks{P}\cup B \setminus \set{\alpha}) \setminus B) = \V(P)
      &
      \notag
      \\
      &
      \qedLocal
      \sources{P} \cup (\sinks{P} \setminus \set{\alpha}) = \V(P)
      &
      \notag
      \\
      &
      \qedLocal
      \sources{P} \cup \sinks{P} = \V(P)
      &
      \text{$\alpha \in \sources{P}$}
      \notag
      \\
      &
      \caseName{(1.V), $\Rightarrow$}
      \notag
      \\
      &
      H,P,G
      \text{ a partial slice}
      &
      \text{suppose}
      \notag
      \\
      &
      \forall (\alpha, \beta) \in \E(G) \cup \E(P).\; \beta \notin \V(H)
      &
      \text{def. partial slice}
      \notag
      \\
      &
      \E(G) \cup \E(P) = \E(G \edgeMinus E) \cup \E(P \cup E)
      &
      \notag
      \\
      &
      \qedLocal
      \forall (\alpha, \beta) \in \E(G \edgeMinus E) \cup \E(P \cup E).\; \beta \notin \V(H)
      &
      \notag
      \\
      &
      \caseName{(1.V), $\Leftarrow$}
      \notag
      \\
      &
      H,P \cup E, G \edgeMinus E \text{ a partial slice}
      &
      \text{suppose}
      \notag
      \\
      &
      \forall (\alpha, \beta) \in \E(G \edgeMinus E) \cup \E(P \cup E).\; \beta \notin \V(H)
      &
      \text{def. partial slice}
      \notag
      \\
      &
      \E(G) \cup \E(P) = \E(G \edgeMinus E) \cup \E(P \cup E)
      &
      \notag
      \\
      &
      \qedLocal
      \forall (\alpha, \beta) \in \E(G) \cup \E(P).\; \beta \notin \V(H)
      &
      \notag
      \\
      &
      \caseName{(1.VI)}
      \notag
      \\
      &
      \text{Suffices to show: }
      \sources{G}\setminus(\sinks{P}\setminus\sources{P}) = \sources{G \edgeMinus E}\setminus(\sinks{P \cup E}\setminus \sources{P \cup E})
      &
      \notag
      \\
      &
      \sources{G}\setminus(\sinks{P}\setminus\sources{P}) = \V(H) \uplus (\sources{G_0}\setminus\sources{H})
      &
      \text{suppose}
      \notag
      \\
      &
      E = \set{\alpha} \times B
      &
      \notag
      \\
      &
      \subcase{x \in \sources{G}\setminus(\sinks{P}\setminus\sources{P})}
      \notag
      \\
      &
      x \in \sources{G} \subseteq \sources{G\edgeMinus E}
      &
      \notag
      \\
      &
      x \notin \sinks{P}\setminus\sources{P}
      &
      \\
      &
      x \notin B
      &
      \text{$x \in \sources{G}$}
      \notag
      \\
      &
      \qedLocal
      x \notin \sinks{P\cup E}\setminus\sources{P\cup E}
      &
      \notag
      \\
      &
      x \in \sources{G \edgeMinus E}\setminus(\sinks{P \cup E}\setminus \sources{P \cup E})
      &
      \notag
      \\
      &
      \qedLocal
      \sources{G}\setminus{\sinks{P}\setminus\sources{P}} \subseteq \sources{G \edgeMinus E}\setminus(\sinks{P \cup E}\setminus \sources{P \cup E})
      &
      \locallabel{unmodified-subset-modified}
      \\
      &
      \subcase{x \notin \sources{G}\setminus(\sinks{P}\setminus\sources{P})}
      \notag
      \\
      &
      \subcase{x \in \sources{G}}
      \notag
      \\
      &
      x \in \sinks{P}\setminus\sources{P}
      &
      \\
      &
      x \in \sinks{P \cup E}\setminus\sources{P \cup E}
      &
      \notag
      \\
      &
      \qedLocal
      x \notin \sources{G\edgeMinus E}\setminus(\sinks{P \cup E}\setminus\sources{P \cup E})
      &
      \notag
      \\
      &
      \subcase{x \notin \sources{G} \wedge x \in \sources{G \edgeMinus E}}
      \notag
      \\
      &
      x \notin \sources{G}\setminus(\sinks{P}\setminus\sources{P})
      &
      \notag
      \\
      &
      (\alpha,x) \in E
      &
      \text{$x \in \sources{G \edgeMinus E}\setminus \sources{G}$}
      \\
      &
      x \in (\sinks{P \cup E} \setminus \sources{P \cup E})
      &
      \notag
      \\
      &
      \qedLocal
      x \notin \sources{G \edgeMinus E} \setminus (\sinks{P \cup E} \setminus \sources{P \cup E})
      &
      \notag
      \\
      &
      \qedLocal
      \sources{G\edgeMinus E}\setminus(\sinks{P \cup E}\setminus\sources{P \cup E}) \subseteq \sources{G}\setminus(\sinks{P}\setminus\sources{P})
      &
      \\
      &
      \qedLocal
      \sources{G}\setminus{\sinks{P}\setminus\sources{P}} = \sources{G \edgeMinus E}\setminus(\sinks{P \cup E}\setminus \sources{P \cup E})
      &
      \notag
      \\
      &
      \caseName{(2.II), $\Rightarrow$}
      \notag
      \\
      &
      E = B \times \set{\alpha}
      &
      \text{suppose}
      \notag
      \\
      &
      \compl{\sinks{P}} \subseteq \V(H)
      &
      \text{suppose}
      \notag
      \\
      &
      \V(H \cup E) = \V(H) \cup \set{\alpha}
      &
      \notag
      \\
      &
      \compl{\sinks{P \edgeMinus E}} \cup B = \compl{\sinks{P}}
      &
      \notag
      \\
      &
      B \subseteq \compl{\sinks{P}} \subseteq \V(H)
      &
      \notag
      \\
      &
      \qedLocal
      \compl{\sinks{P \edgeMinus E}} \subseteq \V(H) \cup \set{\alpha}
      &
      \notag
      \\
      &
      \caseName{(2.II), $\Leftarrow$}
      \notag
      \\
      &
      \compl{\sinks{P \edgeMinus E}} \subseteq \V(H \cup E)
      &
      \text{suppose}
      \notag
      \\
      &
      \compl{\sinks{P}} \subseteq \compl{\sinks{P \edgeMinus E}} \cup B
      &
      \notag
      \\
      &
      \forall \beta \in B.\;\beta \in \V(H) \cup \set{\alpha}
      &
      \text{def. $E$}
      \notag
      \\
      &
      \qedLocal
      \forall \beta \in B.\; \beta \in \V(H)
      &
      \text{$\beta \neq \alpha$}
      \notag
      \\
      &
      \qedLocal
      \compl{\sinks{P \edgeMinus E}} \cup B \subseteq \V(H) \cup \set{\alpha}
      &
      \notag
      \\
      &
      \qedLocal
      \compl{\sinks{P}} \subseteq \V(H)
      &
      \\
      &
      \caseName{(2.III), $\Rightarrow$}
      \notag
      \\
      &
      \V(H) = \sources{P}
      &
      \text{suppose}
      \notag
      \\
      &
      \text{want to show: }
      \V(H \cup E) = \sources{P \edgeMinus E}
      &
      \notag
      \\
      &
      \V(H \cup E) = \V(H) \cup \set{\alpha} = \sources{P} \cup \set{\alpha}
      &
      \notag
      \\
      &
      \sources{P \edgeMinus E} = \sources{P} \cup \set{\alpha}
      &
      \text{def. $E$}
      \notag
      \\
      &
      \qedLocal
      \sources{P \edgeMinus E} = \V(H \cup E)
      &
      \notag
      \\
      &
      \caseName{(2.III), $\Leftarrow$}
      \notag
      \\
      &
      \V(H \cup E) = \sources{P \edgeMinus E}
      &
      \text{suppose}
      \notag
      \\
      &
      \V(H) = \V(H \cup E) \setminus \set{\alpha}
      &
      \notag
      \\
      &
      \sources{P \edgeMinus E} = \sources{P} \cup \set{\alpha}
      &
      \text{def. $E$}
      \notag
      \\
      &
      \qedLocal
      \sources{P \edgeMinus E} \setminus \set{\alpha} = \sources{P}
      &
      \notag
      \\
      &
      \qedLocal
      \V(H) = \sources{P}
      &
      \notag
      \\
      &
      \caseName{(2.IV), $\Rightarrow$}
      \notag
      \\
      &
      \V(P) = \sources{P} \cup \sinks{P}
      &
      \text{suppose}
      \locallabel{vp-equals-sources-sinks}
      \\
      &
      \text{want to show: }
      \V(P) = \sources{P \edgeMinus E} \cup \sinks{P \edgeMinus E}
      &
      \notag
      \\
      &
      \sources{P} \subseteq \sources{P \edgeMinus E} \wedge \sinks{P} \subseteq \sinks{P \edgeMinus E}
      &
      \text{straightforward}
      \notag
      \\
      &
      \text{Remains to show: }
      \sources{P \edgeMinus E} \cup \sinks{P \edgeMinus E} \subseteq \sources{P} \cup \sinks{P}
      &
      \notag
      \\
      &
      x \in \V(P \edgeMinus E)
      &
      \text{suppose}
      \notag
      \\
      &
      \subcase{x \in \sources{P \edgeMinus E}}
      \notag
      \\
      &
      x \in \V(P)
      &
      \notag
      \\
      &
      \qedLocal
      x \in \sources{P} \cup \sinks{P}
      &
      \text{\localref{vp-equals-sources-sinks}}
      \\
      &
      \subcase{x \in \sinks{P \edgeMinus E}}
      \notag
      \\
      &
      x \in \V(P)
      &
      \notag
      \\
      &
      \qedLocal
      x \in \sources{P} \cup \sinks{P}
      &
      \text{\localref{vp-equals-sources-sinks}}
      \notag
      \\
      &
      \caseName{(2.IV), $\Leftarrow$}
      \notag
      \\
      &
      \V(P) = \V(P \edgeMinus E) = \sources{P \edgeMinus E} \cup \sinks{P \edgeMinus E}
      &
      \text{suppose}
      \notag
      \\
      &
      \text{want to show: }
      \sources{P \edgeMinus E} \cup \sinks{P \edgeMinus E} = \sources{P} \cup \sinks{P}
      &
      \notag
      \\
      &
      \sources{P} \cup \sinks{P} \subseteq \sources{P \edgeMinus E} \cup \sinks{P \edgeMinus E}
      &
      \notag
      \\
      &
      \text{remains to show: }
      \sources{P \edgeMinus E} \cup \sinks{P \edgeMinus E} \subseteq \sources{P} \cup \sinks{P}
      &
      \notag
      \\
      &
      x \in \V(P)
      &
      \text{suppose}
      \notag
      \\
      &
      \subcase{x \in \sources{P \edgeMinus E}}
      \notag
      \\
      &
      \sources{P \edgeMinus E}\setminus \sources{P} = \set{\alpha}
      &
      \text{def. $E$}
      \notag
      \\
      &
      \subcase{x = \alpha}
      \notag
      \\
      &
      \alpha \notin \sinks{P}
      &
      \text{assume}
      \notag
      \\
      &
      \alpha \in \V(H)
      &
      \text{partial slice (II)}
      \notag
      \\
      &
      (\beta, \alpha) \in \E(P)
      &
      \text{def. $E$}
      \notag
      \\
      &
      \qedLocal
      \bot
      &
      \text{contra partial slice (V)}
      \\
      &
      \qedLocal
      \alpha \in \sinks{P}
      &
      \notag
      \\
      &
      \qedLocal x \in \sources{P} \cup \sinks{P}
      &
      \notag
      \\
      &
      \subcase{x \in \sources{P}}
      \notag
      \\
      &
      \qedLocal
      x \in \sources{P} \cup \sinks{P}
      &
      \notag
      \\
      &
      \subcase{x \in \sinks{P \edgeMinus E}}
      \notag
      \\
      &
      \sinks{P \edgeMinus E} \setminus \sinks{P} \subseteq B
      &
      \text{def. $E$}
      \notag
      \\
      &
      \subcase{x \in B}
      \notag
      \\
      &
      B \subseteq \V(H \cup E)
      &
      \text{def. $E$}
      \\
      &
      \alpha \notin B
      &
      \notag
      \\
      &
      \qedLocal
      B \in \sources{P}
      &
      \notag
      \\
      &
      \qedLocal
      x \in \sources{P} \subseteq \sources{P} \cup \sinks{P}
      &
      \notag
      \\
      &
      \subcase{x \in \sinks{P}}
      \notag
      \\
      &
      \qedLocal
      x \in \sources{P} \cup \sinks{P}
      &
      \notag
      \\
      &
      \qedLocal
      \sinks{P \edgeMinus E} \subseteq \sources{P} \cup \sinks{P}
      &
      \\
      &
      \caseName{(2.V), $\Rightarrow$}
      \notag
      \\
      &
      H,P,G
      \text{ a partial slice}
      &
      \text{suppose}
      \notag
      \\
      &
      \forall \beta \in \V(G).\;\nexists (\beta, \alpha) \in \E(G)
      &
      \text{$\alpha \in \sources{G}$}
      \locallabel{alpha-source}
      \\
      &
      \forall (x,y) \in \E(G). y \notin \V(H \cup E)
      &
      \text{def. partial slice, \localref{alpha-source}}
      \notag
      \\
      &
      \V(H \cup E) = \V(H) \cup \set{\alpha}
      &
      \text{def. $E$}
      \notag
      \\
      &
      (x, y) \in \E(P)
      &
      \text{suppose}
      \notag
      \\
      &
      \subcase{(x,y) \in \E(P) \edgeMinus E}
      \notag
      \\
      &
      y \notin \V(H)
      &
      \notag
      \\
      &
      y \neq \alpha
      &
      \text{$(x,y) \notin E$}
      \notag
      \\
      &
      \qedLocal
      y \notin \V(H \cup E)
      &
      \notag
      \\
      &
      \subcase{(x,y) \in E}
      \notag
      \\
      &
      \qedLocal
      (x,y) \notin \E(P) \edgeMinus E
      &
      \notag
      \\
      &
      \caseName{(2.V), $\Leftarrow$}
      \notag
      \\
      &
      H \cup E, P \edgeMinus E, G \text{ a partial slice}
      &
      \text{suppose}
      \notag
      \\
      &
      \forall (\alpha, \beta) \in \E(G) \cup \E(P \edgeMinus E).\; \beta \notin \V(H \cup E)
      &
      \text{def. partial slice}
      \notag
      \\
      &
      \forall (\alpha, \beta) \in \E(G) \cup \E(P \edgeMinus E).\; \beta \notin \V(H)
      &
      \text{immediate}
      \notag
      \\
      &
      (\beta, \alpha) \in E \subset \E(P \edgeMinus E) \cup E = \E(P)
      &
      \text{suppose}
      \notag
      \\
      &
      \alpha \notin \V(H)
      &
      \text{def. $E$}
      \notag
      \\
      &
      \qedLocal
      \forall (\alpha', \beta') \in \E(P)
      &
      \notag
      \\
      &
      \qedLocal
      \forall (\alpha', \beta') \in \E(G) \cup \E(P).\; \beta \notin \V(H)
      &
      \notag
      \\
      &
      \caseName{(2.VI)}
      \notag
      \\
      &
      \text{Suffices to show: }
      \sinks{P \edgeMinus E}\setminus \sources{P \edgeMinus E} \cup \set{\alpha} = \sinks{P}\setminus\sources{P}
      &
      \notag
      \\
      &
      E = \inN{P}(\alpha) \times \set{\alpha}
      &
      \notag
      \\
      &
      \subcase{\sinks{P \edgeMinus E}\setminus \sources{P \edgeMinus E} \cup \set{\alpha} \subseteq \sinks{P}\setminus\sources{P}}
      \notag
      \\
      &
      \subcase{x = \alpha}
      \notag
      \\
      &
      \alpha \notin \sources{P}
      &
      \text{$E \neq \emptyset$}
      \notag
      \\
      &
      \alpha \in \sinks{P}
      &
      \text{$P$ has depth 1}
      \notag
      \\
      &
      \qedLocal
      x \in \sinks{P}\setminus\sources{P}
      &
      \notag
      \\
      &
      \subcase{x \neq \alpha \wedge x \in \sinks{P \edgeMinus E}\setminus \sources{P \edgeMinus E}}
      \notag
      \\
      &
      \sinks{P \edgeMinus E}\setminus \sinks{P} \subseteq \inN{P}(\alpha)
      &
      \notag
      \\
      &
      \inN{P}(\alpha) \subseteq \sources{P} \subseteq \sources{P \edgeMinus E}
      &
      \notag
      \\
      &
      \qedLocal
      x \in \sinks{P}
      &
      \text{$x \sinks{P \edgeMinus E} \wedge \in \notin \inN{P}(\alpha)$}
      \notag
      \\
      &
      \sources{P \edgeMinus E} \supseteq \sources{P}
      &
      \notag
      \\
      &
      x \notin \sources{P}
      &
      \text{$x \notin \sources{P \edgeMinus E}$}
      \notag
      \\
      &
      \qedLocal
      x \in \sinks{P}\setminus \sources{P}
      &
      \notag
      \\
      &
      \qedLocal
      \sinks{P \edgeMinus E}\setminus \sources{P \edgeMinus E} \cup \set{\alpha} \subseteq \sinks{P}\setminus\sources{P}
      &
      \notag
      \\
      &
      \subcase{\sinks{P}\setminus\sources{P} \subseteq \sinks{P \edgeMinus E}\setminus \sources{P \edgeMinus E} \cup \set{\alpha}}
      \notag
      \\
      &
      x \in \sinks{P}\setminus\sources{P}
      &
      \text{suppose}
      \notag
      \\
      &
      x \notin \sources{P}
      &
      \\
      &
      x \in \sinks{P} \subseteq \sinks{P \edgeMinus E}
      &
      \notag
      \\
      &
      \subcase{x \notin \sources{P \edgeMinus E}}
      \notag
      \\
      &
      \qedLocal
      x \in \sinks{P \edgeMinus E}\setminus \sources{P \edgeMinus E} \subseteq \sinks{P \edgeMinus E}\setminus \sources{P \edgeMinus E} \cup \set{\alpha}
      &
      \notag
      \\
      &
      \subcase{x \in \sources{P \edgeMinus E}}
      \notag
      \\
      &
      x \in \sources{P \edgeMinus E} \setminus \sources{P}
      &
      \notag
      \\
      &
      \sources{P \edgeMinus E} \setminus \sources{P} = \set{\alpha}
      &
      \notag
      \\
      &
      \qedLocal
      x = \alpha \in \sinks{P \edgeMinus E}\setminus \sources{P \edgeMinus E} \cup \set{\alpha}
      &
      \notag
      \\
      &
      \qedLocal
      \sinks{P}\setminus\sources{P} \subseteq \sinks{P \edgeMinus E}\setminus \sources{P \edgeMinus E} \cup \set{\alpha}
      &
      \notag
      \\
      &
      \qedLocal
      \sinks{P}\setminus\sources{P} = \sinks{P \edgeMinus E}\setminus \sources{P \edgeMinus E} \cup \set{\alpha}
      &
      \notag
   \end{flalign}
\end{proof}

\begin{lemma}
   \label{lem:appendix:proof:conjugate:algorithms-suff:aux}
   Suppose $H,P,G$ a partial slice of $G_0$, then:
   \begin{enumerate}
      \item $\sufficesR^*_{G \cup P \cup H}(\sources{H}) = \sufficesR^*_{(G \edgeMinus E)\cup (P \cup E) \cup H}(\sources{H})$
      \item $\sufficesR^*_{G \cup P \cup H}(\sources{H}) = \sufficesR^*_{G \cup (P \edgeMinus E) \cup (H \cup E)}(\sources{H \cup E})$
   \end{enumerate}
\end{lemma}
\proofContext{auxiliary-lemma}
\begin{proof}
   \begin{enumerate}
      \item $\Lowlight{G_0 = }\; G \cup P \cup H = (G \edgeMinus E) \cup (P \cup E) \cup H$
      \item $\Lowlight{G_0 = }\; G \cup P \cup H = G \cup (P \edgeMinus E) \cup (H \cup E)$, and $\sources{H} = \sources{H \cup E}$
   \end{enumerate}
\end{proof}
\begin{proposition}
   \label{prop:appendix:proof:conjugate:algorithms-suff:correctness-sufficesE}
   Suppose $H,P,G$ a partial slice of $G_0$:
   \[ (\exists H'. \; H,P,G \sufficesE{} H' \text{ with }\V(H') = Y) \iff \sufficesR^*_{G_0}(\sources{H}) = Y\]
\end{proposition}

\proofContext{correctness-sufficesE}
\begin{proof}
   \small
   \begin{flalign}
      &
      H, P, G
      \text{ a partial slice}
      &
      \text{suppose}
      \\
      &
      \caseName{$\Rightarrow$ direction:}
      \notag
      \\
      &
      H, P, G \sufficesE{} H' \text{ with } \V(H')=Y
      &
      \text{suppose, $\exists H'$}
      \notag
      \\
      &
      \text{By induction on derivation of $\sufficesE{}$.}
      &
      \notag
      \\
      &
      \subcaseDerivation{\derivationWidth}{
         \begin{smathpar}
            \inferrule*[
               lab={\ruleName{done}},
               right={$\sources{G} \cap \V(P) \subseteq \sources{P} \wedge \V(H) \subseteq \sinks{G}$}
            ]
            {
               \strut
            }
            {
               H, P, G \sufficesE{} H
            }
         \end{smathpar}
      }
      &
      \notag
      \\
      &
      \qedLocal
      \sufficesR^*_{G_0}(\sources{H}) = \V(H)
      &
      \text{\lemref{appendix:proof:conjugate:algorithms-suff:suffices-sources}}
      \notag
      \\
      &
      \subcaseDerivation{\derivationWidth}{
         \begin{smathpar}
            \inferrule*[
               lab={\ruleName{pending}},
               right={$\alpha \in \V(H) \wedge E = \outE{G}(\alpha) \neq \emptyset$}
            ]
            {
               H, P \cup E, G \edgeMinus E
               \sufficesE{}
               H'
            }
            {
               H, P, G
               \sufficesE{}
               H'
            }
         \end{smathpar}
      }
      &
      \notag
      \\
      &
      H, P \cup E, G \edgeMinus E
      \text{ a partial slice of } G_0
      &
      \text{\lemref{appendix:proof:conjugate:algorithms-suff:partial-slice-preserved}}
      \notag
      \\
      &
      \sufficesR^*_{(G \edgeMinus E) \cup (P \cup E) \cup H}(\sources{H}) = Y
      &
      \text{IH}
      \notag
      \\
      &
      \qedLocal
      \sufficesR^*_{G_0}(\sources{H}) = Y
      &
      \text{\lemref{appendix:proof:conjugate:algorithms-suff:aux}(i)}
      \notag
      \\
      &
      \subcaseDerivation{\derivationWidth}{
         \begin{smathpar}
            \inferrule*[
               lab={\ruleName{extend}},
               right={$\alpha \in \sources{G} \wedge E = \inE{P}(\alpha) \neq \emptyset$}
            ]
            {
               H \cup E,
               P \edgeMinus E,
               G
               \sufficesE{}
               H'
            }
            {
               H, P, G
               \sufficesE{}
               H'
            }
         \end{smathpar}
      }
      &
      \notag
      \\
      &
      H \cup E, P \edgeMinus E, G
      \text{ a partial slice of } G_0
      &
      \text{\lemref{appendix:proof:conjugate:algorithms-suff:partial-slice-preserved}}
      \notag
      \\
      &
      \sufficesR^*_{G \cup (P \edgeMinus E) \cup (H \cup E)}(\sources{H \cup E})
      &
      \text{IH}
      \notag
      \\
      &
      \qedLocal
      \sufficesR^*_{G_0}(\sources{H})
      &
      \text{\lemref{appendix:proof:conjugate:algorithms-suff:aux}(ii)}
      \notag
      \\
      &
      \caseName{$\Leftarrow$ direction}
      \notag
      \\
      &
      \text{By strong lexicographical induction over $(\length{\E(G)},\length{\E(P)})$}
      &
      \notag
      \\
      &
      \sufficesR^*_{G_0}(\sources{H}) = Y
      &
      \text{suppose}
      \notag
      \\
      &
      \subcase{\sources{G} \cap \V(P) \subseteq \sources{P} \wedge \V(H) \subseteq \sinks{G}}
      \notag
      \\
      &
      \derivation{\derivationWidth}{
         \begin{smathpar}
            \inferrule*[
               lab={\ruleName{done}},
               right={$\sources{G} \cap \V(P) \subseteq \sources{P} \wedge \V(H) \subseteq \sinks{G}$}
            ]
            {
               \strut
            }
            {
               H, P, G \sufficesE{} H
            }
         \end{smathpar}
      }
      &
      \text{def. $\sufficesE{}$}
      \notag
      \\
      &
      \qedLocal
      \sufficesR^*_{G_0}(\sources{H}) = \V(H) =  Y
      &
      \text{\lemref{appendix:proof:conjugate:algorithms-suff:suffices-sources}}
      \notag
      \\
      &
      \subcase{\alpha \in \V(H) \wedge E = \outE{G}(\alpha) \neq \emptyset}
      &
      \text{$\exists \alpha$, $\exists E$}
      \notag
      \\
      &
      (\length{\E(G \edgeMinus E)}, \length{\E(P)} )< (\length{\E(G)}, \length{\E(P \cup E)})
      &
      \notag
      \\
      &
      H, P \cup E, G \edgeMinus E
      \text{ a partial slice of } G_0
      &
      \text{\lemref{appendix:proof:conjugate:algorithms-suff:partial-slice-preserved}}
      \notag
      \\
      &
      H, P \cup E, G \edgeMinus E \sufficesE{} H'
      &
      \text{IH; $\exists H'$}
      \notag
      \\
      &
      \derivation{\derivationWidth}{
         \begin{smathpar}
            \inferrule*[
               lab={\ruleName{pending}},
               right={$\alpha \in \V(H) \wedge E = \outE{G}(\alpha) \neq \emptyset$}
            ]
            {
               H, P \cup E, G \edgeMinus E
               \sufficesE{}
               H'
            }
            {
               H, P, G
               \sufficesE{}
               H'
            }
         \end{smathpar}
      }
      &
      \text{def. $\sufficesE{}$}
      \notag
      \\
      &
      \sufficesR^*_{(G \edgeMinus E) \cup (P \cup E) \cup H}(\sources{H}) = \sufficesR^*_{G \cup P \cup H}(\sources{H})
      &
      \text{\lemref{appendix:proof:conjugate:algorithms-suff:aux}(i)}
      \notag
      \\
      &
      \qedLocal
      \sufficesR^*_{G_0}(\sources{H}) = \V(H')\;\Lowlight{ = Y}
      &
      \text{IH}
      \notag
      \\
      &
      \subcase{\alpha \in \sources{G} \wedge E = \inE{P}(\alpha) \neq \emptyset}
      &
      \text{$\exists \alpha$, $\exists E$}
      \notag
      \\
      &
      (\length{\E(G)},\length{\E(P \edgeMinus E)}) < (\length{\E(P)},\length{\E(P)})
      &
      \notag
      \\
      &
      H \cup E, P \edgeMinus E, G
      \text{ a partial slice of } G_0
      &
      \text{\lemref{appendix:proof:conjugate:algorithms-suff:partial-slice-preserved}}
      \notag
      \\
      &
      H \cup E, P \edgeMinus E, G \sufficesE{} H'
      &
      \text{IH, $\exists H'$}
      \notag
      \\
      &
      \derivation{\derivationWidth}{
         \begin{smathpar}
            \inferrule*[
               lab={\ruleName{extend}},
               right={$\alpha \in \sources{G} \wedge E = \inE{P}(\alpha) \neq \emptyset$}
            ]
            {
               H \cup E,
               P \edgeMinus E,
               G
               \sufficesE{}
               H'
            }
            {
               H, P, G
               \sufficesE{}
               H'
            }
         \end{smathpar}
      }
      &
      \text{def. $\sufficesE{}$}
      \notag
      \\
      &
      \sufficesR^*_{G \cup P \cup H}(\sources{H}) = \sufficesR^*_{G \cup (P \edgeMinus E) \cup (H \cup E)}(\sources{H \cup E})
      &
      \text{\lemref{appendix:proof:conjugate:algorithms-suff:aux}(ii)}
      \\
      &
      \qedLocal
      \sufficesR^*_{G_0}(\sources{H}) = \V(H')\;\Lowlight{ = Y}
      &
      \text{IH}
      \notag
   \end{flalign}
\end{proof}

\section{Source code for data visualisation figures and benchmarked programs}
\label{sec:apdx:evaluation-src-code}

\begin{figure}[h]
    \small
    {
\begin{lstlisting}[language=Fluid]
let isCountry name x = name == x.country;
    isYear year x = year == x.year;

let plot year countries =
  let rens = filter (isYear year) renewables;
      nonRens = filter (isYear year) nonRenewables;
  let plotCountry country =
    let rens' = filter (isCountry country) rens;
        rensOut = sum (map (fun x $\rightarrow$ x.output) rens');
        rensCap = sum (map (fun x $\rightarrow$ x.capacity) rens');
        x = head (filter (isCountry country) nonRens);
        nonRensCap = x.nuclearCap + x.petrolCap + x.gasCap + x.coalCap
    in {
      x: rensCap / (rensCap + nonRensCap),
      y: (rensOut + x.nuclearOut) / (rensCap + x.nuclearCap)
    }
  in map plotCountry countries

in ScatterPlot {
  caption: "Clean energy efficiency vs. proportion of renewable energy capacity",
  data: plot 2018 [ "BRA", "CHN", "DEU", "FRA", "EGY", "IND", "JPN", "MEX", "NGA", "USA" ],
  xlabel: "Renewables/TotalEnergyCap",
  ylabel: "Clean Capacity Factor"
}
\end{lstlisting}}
  \caption{Source code: \lstinline{energy-query} (\figref{introduction:related-inputs} that visualises linked inputs)}
\end{figure}
\begin{figure}
    \small
    {
\begin{lstlisting}[language=Fluid]
MultiPlot $\lbbar$
   "stacked-bar-chart" :=
      let totalFor year country rows =
         let [ row ] = [ row | row $\leftarrow$ rows, row.year == year, row.country == country ]
         in row.nuclearOut + row.gasOut + row.coalOut + row.petrolOut;
      let stack year = [ { y: country, z: totalFor year country nonRenewables }
                     | country $\leftarrow$ ["BRA", "EGY", "IND", "JPN"] ]
      in BarChart {
         caption: "Non-renewables by country",
         data: [ { x: numToStr year, bars: stack year }
               | year $\leftarrow$ [2014..2018] ]
      },
   "scatter-plot" :=
      let isCountry name x = name == x.country;
         isYear year x = year == x.year;

      let plot year countries =
      let rens = filter (isYear year) renewables;
            nonRens = filter (isYear year) nonRenewables;
      let plotCountry country =
         let rens' = filter (isCountry country) rens;
            rensOut = sum (map (fun x $\rightarrow$ x.output) rens');
            rensCap = sum (map (fun x $\rightarrow$ x.capacity) rens');
            x = head (filter (isCountry country) nonRens);
            nonRensCap = x.nuclearCap + x.petrolCap + x.gasCap + x.coalCap
         in {
            x: rensCap / (rensCap + nonRensCap),
            y: (rensOut + x.nuclearOut) / (rensCap + x.nuclearCap)
         }
      in map plotCountry countries

      in ScatterPlot {
         caption: "Clean energy efficiency vs proportion of renewable energy capacity",
         data: plot 2018 [ "BRA", "CHN", "DEU", "FRA", "EGY"
                         , "IND", "JPN", "MEX", "NGA", "USA" ],
         xlabel: "Renewables/TotalEnergyCap",
         ylabel: "Clean Capacity Factor"
      }
$\rbbar$
\end{lstlisting}}
  \caption{Source code: \lstinline{stacked-bar-scatter-plot} (\figref{introduction:scatterplot} that visualises linked outputs)}
\end{figure}

\begin{figure}[h]
  \begin{subfigure}[position]{1\textwidth}
    \small
    {
\begin{lstlisting}[language=Fluid]
let filter =
   let edgeDetect = [[0,  1, 0],
                     [1, -4, 1],
                     [0,  1, 0]] in
   $\llbracket$ nth2 i j edgeDetect | (i, j) in (3, 3) $\rrbracket$;
convolve input_image filter extend
\end{lstlisting}}
    \caption{\lstinline{edge-detect}}
  \end{subfigure}
  \begin{subfigure}[position]{1\textwidth}
    \small
    {
\begin{lstlisting}[language=Fluid]
let filter =
   let emboss = [[-2, -1, 0],
                 [-1,  1, 1],
                 [ 0,  1, 2]] in
   $\llbracket$ nth2 i j emboss | (i, j) in (3, 3) $\rrbracket$;
convolve input_image filter zero
\end{lstlisting}}
    \caption{\lstinline{emboss}}
  \end{subfigure}
  \begin{subfigure}[position]{1\textwidth}
    \small
    {
\begin{lstlisting}[language=Fluid]
let filter =
   let gaussian = [[1,  4, 1],
                   [4, 16, 4],
                   [1,  4, 1]] in
   $\llbracket$ nth2 i j gaussian | (i, j) in (3, 3) $\rrbracket$;
convolve input_image filter zero
\end{lstlisting}}
    \caption{\lstinline{gaussian}}
  \end{subfigure}
  \begin{subfigure}[position]{\textwidth}
    \small
    {
\begin{lstlisting}[language=Fluid]
let zero n = const n;
    extend n = min (max n 1);
let convolve image kernel method =
    let ((m, n), (i, j)) = (dims image, dims kernel);
        (half_i, half_j) = (i `quot` 2, j `quot` 2);
        area = i * j
    in  $\llbracket$ let weightedSum = sum [
           image!(x, y) * kernel!(i' + 1, j' + 1)
           | (i', j') $\leftarrow$ range (0, 0) (i - 1, j - 1),
             let x = method (m' + i' - half_i) m,
             let y = method (n' + j' - half_j) n,
             x >= 1, x <= m, y >= 1, y <= n
         ] in weightedSum `quot` area
         | (m', n') in (m, n) $\rrbracket$;
let input_image =
   let image = [[15, 13,  6, 9, 16],
                 [12,  5, 15, 4, 13],
                 [14,  9, 20, 8,  1],
                 [ 4, 10,  3, 7, 19],
                 [ 3, 11, 15, 2,  9]] in
   $\llbracket$ nth2 i j image | (i, j) in (5, 5) $\rrbracket$;
\end{lstlisting}}
    \caption{Library code}
  \end{subfigure}
  \caption{Source code: convolution examples}
\end{figure}

\begin{figure}[h]
  \centering
  \small
  \begin{subfigure}[position]{1\textwidth}
    \small
    {
\begin{lstlisting}[language=Fluid]
let year = 2015;
let exclude countries yearData =
   flip map yearData
            (second (filter (fun (country, countryData) $\rightarrow$ not (elem country countries))))
in caption ("Renewables (GW) by country and energy type, " ++ numToStr year)
   (groupedBarChart True colours1 0.2
      (exclude [] (lookup year data)))
\end{lstlisting}}
    \caption{\lstinline{grouped-bar-chart}}
  \end{subfigure}
  \begin{subfigure}[position]{1\textwidth}
    \small
    {
\begin{lstlisting}[language=Fluid]
let year1 = head data;
let addTotal kvs =
   ("Total", sum (map snd kvs)) : kvs;
let countryData country =
   map (second (compose addTotal (lookup country))) data in
(fun country $\rightarrow$
   caption country
      (lineChart True ("black" : colours1) (fst year1)
         (countryData country)))
"China"
\end{lstlisting}}
    \caption{\lstinline{line-chart}}
  \end{subfigure}
  \begin{subfigure}[position]{1\textwidth}
    \small
    {
\begin{lstlisting}[language=Fluid]
let year = 2015 in
caption ("Total renewables (GW) by country, " ++ numToStr year)
   (stackedBarChart True colours1 0.2 (lookup year data))
\end{lstlisting}}
    \caption{\lstinline{stacked-bar-chart}}
  \end{subfigure}
  \caption{Source code: graphics examples}
\end{figure}

\begin{figure}\ContinuedFloat
  \small
\begin{lstlisting}[language=Fluid]
let coords (Group gs) = Point 0 0;
    coords (Rect x y _ _ _) = Point x y;
    coords (Text x y _ _ _) = Point x y;
    coords (Viewport x y _ _ _ _ _ _ _) = Point x y;
let get_x g = let Point x _ = coords g in x;
let get_y g = let Point _ y = coords g in x;
let set_x x (Group gs) = error "Group has immutable coordinates";
    set_x x (Rect _ y w h fill) = Rect x y w h fill;
    set_x x (Text _ y str anchor baseline) = Text x y str anchor baseline;
    set_x x (Viewport _ y w h fill margin scale translate g) =
                Viewport x y w h fill margin scale translate g;
let dimensions2 (Point x1 y1, Point x2 y2) = Point (max x1 x2) (max y1 y2);
let dimensions (Group gs) = foldl (curry dimensions2) (Point 0 0) (map (coords_op) gs);
    dimensions (Polyline ps _ _) = foldl (curry dimensions2) (Point 0 0) ps;
    dimensions (Rect _ _ w h _) = Point w h;
    dimensions (Text _ _ _ _ _) = Point 0 0;
    dimensions (Viewport _ _ w h _ _ _ _ _) = Point w h;
let coords_op g =
        let (Point x y, Point w h) = prod coords dimensions g in
        Point (x + w) (y + h);
let width g = let Point w _ = dimensions g in w;
let height g = let Point _ h = dimensions g in h;
let spaceRight z sep gs =
   zipWith set_x (iterate (length gs) ((+) sep) z) gs;
let colours1 = [ "#66c2a5", "#a6d854", "#ffd92f", "#e5c494"
               , "#fc8d62", "#b3b3b3", "#8da0cb", "#e78ac3"];
let colours2 = [ "#e41a1c", "#377eb8", "#4daf4a", "#984ea3"
               , "#ff7f00", "#ffff33", "#a65628", "#f781bf"];
let scaleToWidth w (Rect x y _ h fill) = Rect x y w h fill;
    scaleToWidth w (Viewport x y w0 h fill margin (Scale x_scale y_scale) translate g) =
      let scale = Scale (x_scale * w / w0) y_scale in
      Viewport x y w h fill margin scale translate g;
let stackRight sep gs =
   map (scaleToWidth (1 - sep)) (spaceRight (sep / 2) 1 gs);
let groupRight sep gs =
   Viewport 0 0 (length gs) (maximum (map height gs)) "none" 0 (Scale 1 1)
           (Translate 0 0) (Group (stackRight sep gs));
let tickEvery n =
   let m = floor (logBase 10 n) in
   if n <= 2 * 10 ** m then 2 * 10 ** (m - 1) else 10 ** m;
\end{lstlisting}
  \caption{Source code: library code for graphics examples (1 of 3)}
\end{figure}
\begin{figure}\ContinuedFloat
  \small
\begin{lstlisting}[language=Fluid]
let axisStrokeWidth = 0.5;
    axisColour = "black";
    backgroundColour = "white";
    defaultMargin = 24;
    markerRadius = 3.5;
    tickLength = 4;
let tick Horiz colour len = Line (Point 0 0) (Point 0 (0 - len)) colour axisStrokeWidth;
    tick Vert colour len = Line (Point 0 0) (Point (0 - len) 0) colour axisStrokeWidth;
let label Horiz x distance str = Text x (0 - distance - 4) str "middle" "hanging";
    label Vert x distance str = Text (0 - distance) x str "end" "central";
let labelledTick orient colour len str =
   Group [tick orient colour len, label orient 0 len str];
let mkPoint Horiz x y = Point y x;
    mkPoint Vert x y = Point x y;
let axis orient x start end =
   let tickSp = tickEvery (end - start);
       firstTick = ceilingToNearest start tickSp;
       lastTick = ceilingToNearest end tickSp;
       n = floor ((end - firstTick) / tickSp) + 1;
       ys = iterate n ((+) tickSp) firstTick;
       ys = if firstTick > start then start : ys else ys;
       ys = if lastTick > end then concat2 ys [lastTick] else ys;
       ps = map (mkPoint orient x) ys;
       ax = Group [
          Line (head ps) (last ps) axisColour axisStrokeWidth,
          Polymarkers ps (flip map ys (compose (labelledTick orient axisColour tickLength)
                          numToStr))
       ]
   in (ax, lastTick);
let catAxis orient x catValues =
   let ys = iterate (length catValues + 1) ((+) 1) 0;
       ps = map (mkPoint orient x) ys
   in Group [
      Line (head ps) (last ps) axisColour axisStrokeWidth,
      Polymarkers (tail ps) (map (const (tick orient axisColour tickLength)) catValues),
      Polymarkers (flip map (tail ps) (fun (Point x y) $\rightarrow$ Point (x - 0.5) y))
                  (map (label orient -0.5 0) catValues)
   ];
\end{lstlisting}
  \caption{Source code: library code for graphics examples (2 of 3)}
\end{figure}
\begin{figure}\ContinuedFloat
  \small
\begin{lstlisting}[language=Fluid]
let viewport x_start x_finish y_finish margin gs =
   Viewport 0 0 (x_finish - x_start) y_finish backgroundColour margin
            (Scale 1 1) (Translate (0 - x_start) 0) (Group gs);
let lineChart withAxes colours x_start data =
   let xs = map fst data;
       nCat = length (snd (head data));
   let plot (n, colour) =
      let ps = map (fun (x, kvs) $\rightarrow$ Point x (snd (nth n kvs))) data
      in Group [
         Polyline ps colour 1,
         Polymarkers ps (repeat (length ps) (Circle 0 0 markerRadius colour))
      ];
   let lines = zipWith (curry plot) (iterate nCat ((+) 1) 0) colours;
       x_finish = last xs;
       y_finish = maximum (flip map data (fun (_, kvs) $\rightarrow$ maximum (map snd kvs)))
   in if withAxes
      then
         let (x_axis, x_finish) = axis Horiz 0 x_start x_finish;
             (y_axis, y_finish') = axis Vert x_start 0 y_finish
         in viewport x_start x_finish y_finish' defaultMargin (x_axis : y_axis : lines)
      else viewport x_start x_finish y_finish 0 lines;
let categoricalChart plotValue withAxes colours sep data =
   let gs = stackRight sep (plotValue colours (map snd data));
       w = length gs;
       h = maximum (map height gs)
   in if withAxes
      then
         let x_axis = catAxis Horiz 0 (map fst data);
             (y_axis, h') = axis Vert 0 0 h
         in viewport 0 w h' defaultMargin (concat2 gs [x_axis, y_axis])
      else viewport 0 w h 0 gs;
let rects colours ns = zipWith (fun colour n $\rightarrow$ Rect 0 0 1 n colour) colours ns;
let stackedBar colours ns =
   let heights = map snd ns;
       subtotals = scanl1 (+) 0 heights;
       dims = zip (0 : subtotals) heights;
       rects = map
          (fun ((y, height), colour) $\rightarrow$ Rect 0 y 1 height colour)
          (zip dims colours)
   in Viewport 0 0 1 (last subtotals) "none" 0 (Scale 1 1) (Translate 0 0) (Group rects);
let barChart = categoricalChart rects;
let groupedBarChart = categoricalChart (compose map (flip (barChart False) 0));
let stackedBarChart = categoricalChart (compose map stackedBar);
let caption str (Viewport x y w h fill margin scale translate g) =
   let g' = Group [ Text (x + w / 2) -2 str "middle" "hanging",
                     Viewport 0 0 w h fill margin scale translate g ]
   in Viewport x y w h backgroundColour (defaultMargin / 2 + 4) (Scale 1 1) (Translate 0 0) g';
\end{lstlisting}
  \caption{Source code: library code for graphics examples (3 of 3)}
\end{figure}

\end{document}